%% file: 0MAIN.tex
\documentclass[nameyear,draft]{econsocart}
\RequirePackage[colorlinks,citecolor=blue,linkcolor=blue,urlcolor=blue,pagebackref]{hyperref}
\usepackage[usenames,dvipsnames]{xcolor}
\usepackage[fontsize=11.6pt]{fontsize}
\usepackage{tikz} \usetikzlibrary{calc, arrows.meta, intersections, patterns, positioning, shapes.misc, fadings, through,decorations.pathreplacing}
\definecolor{ColorOne}{named}{MidnightBlue}
\definecolor{ColorTwo}{named}{Dandelion}
\definecolor{ColorThree}{named}{Plum}

\usepackage{bm}
\usepackage{enumitem,kantlipsum}
\usepackage{accents}

\usepackage{amsthm}
\usepackage{amsmath}
\usepackage{mathtools}
\usepackage{float}
\DeclareSymbolFontAlphabet{\amsmathbb}{AMSb}
\usepackage{xpatch}
\usepackage{amsfonts}
\newcommand{\ub}[1]{\mkern 1.5mu\underline{\mkern-1.5mu#1\mkern-1.5mu}\mkern 1.5mu}
\usepackage{array}
\newtheorem{axiom}{Axiom}
\newtheorem{definition}{Definition}[section]
\newtheorem{corollary}{Corollary}[section]
\newtheorem{theorem}{Theorem}[section]
\newtheorem{proposition}{Proposition}[section]
\newtheorem{lemma}{Lemma}[section]
\newtheorem{assumption}{Assumption}
\theoremstyle{remark}
\newtheorem{experiment}{Thought Experiment}

\startlocaldefs

\endlocaldefs

\begin{document}

\begin{frontmatter}

\title{(Non-)Commutative Aggregation}

\begin{aug}
%
%
%
\author[id=au1,addressref={add1}]{\fnms{Yuzhao}~\snm{Yang}\ead[label=e1]{allenyyz@bu.edu}}

\address[id=add1]{%
\orgdiv{Department of Economics},
\orgname{Boston University}}
\end{aug}
\support{I thank Barton Lipman, Jawwad Noor, and Larry Epstein for their guidance and support. I thank Andrew Ellis, Norio Takeoka, Peter Klibanoff, Marciano Siniscalchi, and the audiences in Science of Decision Making (the University of Hong Kong) for the valuable discussions.}
\begin{abstract}
\textit{Commutativity} is a normative criterion of aggregation and updating stating that the aggregation of expert posteriors should be identical to the update of the aggregated priors. I propose a thought experiment that raises questions about the normative appeal of Commutativity. I propose a weakened version of Commutativity and show how that assumption plays central roles in the characterization of linear belief aggregation, multiple-weight aggregation, and an aggregation rule which can be viewed as the outcome of a game played by ``dual-selves,'' Pessimism and Optimism. Under suitable conditions, I establish equivalences between various relaxations of Commutativity and classic axioms for decision-making under uncertainty, including \text{Independence}, \text{C-Independence}, and \text{Ambiguity Aversion}. 

\end{abstract}



\end{frontmatter}

%




\input{Sections/1}
\input{Sections/2}
\input{Sections/3}

\input{Sections/5}
\input{Sections/6}

\newpage

\bibliographystyle{te} 
\bibliography{myreferences}  

\newpage
\begin{appendix}
\counterwithin{figure}{section}
\counterwithin{table}{section}
\input{Appendix/a.tex}

\input{Appendix/b.tex}

\input{Appendix/c.tex}
\end{appendix}
\end{document}

%% file: Sections/1.tex
\section{Introduction}
\label{sec1}

A central question in the belief aggregation literature is to identify normatively appealing criteria for aggregating individual judgments. One candidate for such criteria states that the aggregation and updating of beliefs should \textit{commute}. That is, if all individuals observe a certain event, then aggregating individual posteriors (\textit{update-then-aggregate}) should produce the same outcome as updating the aggregation of the priors (\textit{aggregate-then-update}). While this condition goes by different names in the literature,\footnote{In the economics literature, Commutativity was named \textit{Independence of Irrelevant Alternatives} by \cite{molavi2018theory} and \textit{Dynamic Rationality} by \cite{dietrich2021fully}. For the problem of selecting beliefs from a set of possible priors, a variant of Commutativity was named \textit{Bayesian Consistency} by \cite{chambers2010bayesian}. In the statistic literature, the names for Commutativity include \textit{External Bayesianity} \citep{madansky1964externally}, \textit{Data Independence Preservation} \citep{mcconway1978combination}, \textit{Prior-to-Posterior Coherence} \citep{weerahandi1978pooling}, and \textit{Weak Likelihood Ratio Axiom} \citep{bordley1982multiplicative}.} I refer to it as \textit{Commutativity}. Commutativity is considered a principle for rational belief aggregation, with its various formulations attracting growing attention in the economics literature \citep{chambers2010bayesian, molavi2018theory, dietrich2021fully}.


The current paper argues that Commutativity is \textit{not} a normatively appealing criterion for the aggregation of heterogeneous priors. To illustrate this point, consider the timeline in Figure \ref{fig1}. A decision-maker (DM) aggregates the beliefs suggested by $n$ experts to formulate a preference relation. For example, the preference may be represented by subjective expected utility (SEU), in which case the DM essentially aggregates the experts' beliefs to form a subjective belief. The current paper considers a more general class of preferences, which allows the DM to perceive ambiguity. The experts start with their prior beliefs at $t=0$, an event $E$ is realized at $t=1$, and the experts update in response to $E$ at $t=2$. 

\begin{figure}[h]
\centering
\input{Sections/figure1.tikz}
\caption{Timeline for Aggregation and Updating}
\label{fig1}
\end{figure}
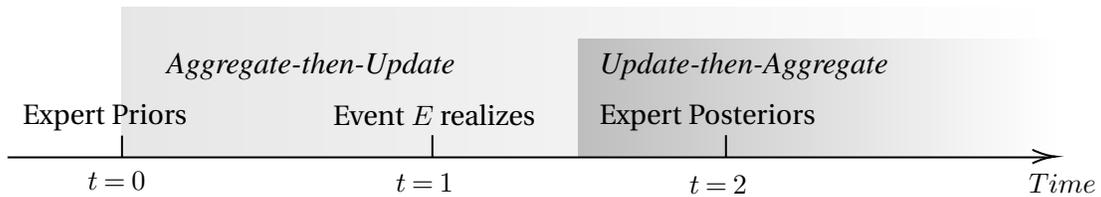

In the \textit{aggregate-then-update} approach, the DM observes the expert priors at $t=0$ and the event realizes at $t=1$, which are sufficient statistics for the expert posteriors at $t=2$. However, in the \textit{update-then-aggregate} approach, the aggregated belief is only dependent on the expert posteriors at $t=2$, which corresponds to the situation where the expert priors at $t=0$ are not observable. Such missing information about expert priors influences belief aggregation. I introduce a thought experiment to illustrate this idea concretely.

\begin{experiment} \label{exp1} (Missing Information about Credibility) A DM receives opinions about the severity of a disease outbreak ($\Omega=\{\text{No, Mild, Severe}\}$) from two experts, Alice and Bob. The prior beliefs of Alice and Bob are respectively $\mu_A=(.9,.1,0)$ and $\mu_B=(0,0,1)$. That is, Alice thinks most likely there is no outbreak, while Bob thinks there is surely a significant outbreak. Next, event $E=\{\text{Mild, Severe}\}$ is observed by both experts and the corresponding posteriors are $\mu_A(\cdot\mid E)=(0,1,0)$ and $\mu_{B}(\cdot\mid E)=(0,0,1)$.

In the \textit{update-then-aggregate} approach, the DM observes and aggregates $\mu_A(\cdot\mid E)$ and $\mu_{B}(\cdot\mid E)$, which are certain about ``Mild'' and ``Severe'' respectively. If the DM deems Alice and Bob as similarly credible, she should be equally confident about the two states. However, in \textit{aggregate-then-update}, when aggregating expert priors, the DM understands that Alice only attaches probability $.1$ to event $E$, so the realization of $E$ will suggest that Alice is a poor predictor. Therefore, conditional on $E$, the result of aggregating the priors should be less confident in the state ``Mild'' believed by Alice. This violates Commutativity. This behavior is reminiscent of the use of maximum likelihood updating \citep{gilboa1993updating} of a set of priors, where the DM discards the priors that assign a low probability to the realized event. 
\end{experiment}

I propose \textit{Weak Commutativity}, which tackles the concern raised by Thought Experiment \ref{exp1}. According to this axiom, the updating and aggregation operations commute \textit{if} the expert priors agree on the probability of every state outside the realized event $E$. In this case, the experts must have a consensus on the probability of $E$ \textit{ex-ante}, so the realization of $E$ no longer reveals the credibility of the experts. This eliminates the concern raised by Thought Experiment \ref{exp1}. 

In addition to tackling the thought experiment, the constraint imposed by Weak Commutativity also facilitates the definition of the update of the DM's preference formulated by aggregating expert suggestions. The current paper assumes the DM aggregates the beliefs suggested by $n$ experts to form a monotone and continuous preference, which allows her to perceive ambiguity due to disagreement among the experts.\footnote{
Ambiguity arises when the DM finds it difficult to aggregate conflicting information from various sources \citep{einhorn1985ambiguity,viscusi1992bayesian,cabantous2007ambiguity}. This is well-known in the literature of modeling preference aggregation \citep{cres2011aggregation, hill2011unanimity, nascimento2012ex,gajdos2013decisions, alon2016utilitarian,stanca2021smooth}.} However, the definition of updating ambiguity-sensitive preferences remains a topic of ongoing debate within the literature \citep{al2009ambiguity,siniscalchi2009two}, which poses difficulties for defining Commutativity. In Section \ref{sec31}, I show that in the situation where Weak Commutativity applies, the set of expert suggestions $\{\mu_1,...,\mu_n\}$ must be \textit{rectangular} with respect to the partition $\{E,E^c\}$ in the sense of \cite{epstein2003recursive}. Under \text{rectangularity}, the canonical dynamic consistency condition \citep{ghirardato2002revisiting} guarantees well-defined conditional preferences for MEU and more general classes of preferences \citep{chandrasekher2022dual}, and the prior-by-prior and maximum likelihood updating of a set of beliefs coincide. As a result, under the constraint imposed by Weak Commutativity, the definition of updating the DM's preference should be uncontroversial.

The main contribution of this paper is twofold. First, I characterize the aggregation rules that satisfy Weak Commutativity. Second, I present a series of propositions that establish connections between (non-)commutativity axioms of belief aggregation and the axioms governing decision-making under uncertainty, including \textit{Independence}, \textit{C-Independence} and \textit{Uncertainty Aversion}. Throughout the paper, I make two unanimity assumptions about the aggregation rule. \textit{Pareto} is analogous to the classic unanimity condition proposed by \cite{harsanyi1955cardinal}. \textit{Monotonicity} states the DM's utility evaluation of an Anscombe-Aumann act $f$ increases if we change the experts' beliefs so that the expected utility of $f$ rises for all experts.

I first introduce two important special cases of the aggregation rules. First, suppose the DM adheres to Savage's P2 (sure-thing principle), which rules out behaviors akin to \cite{ellsberg1961risk}. In this case, Weak Commutativity characterizes linear opinion pooling \citep{stone1961opinion, degroot1974reaching}. The DM has a \textit{credibility weight} (weight), which is a distribution $\lambda=(\lambda_1,\lambda_2,...,\lambda_n)$ over the set of experts. For $i=1,2,...,n$, the number $\lambda_i\in[0,1]$ represents the confidence assigned to expert $i$. If the vector of beliefs of the experts is $\bm{\mu}=(\mu_1,...,\mu_n)$, the DM calculates the (linear) aggregated belief $p_{\lambda}(\bm{\mu})=\sum_{i=1}^n \lambda_i \mu_i$, and the corresponding utility evaluation of an act $f$ is equal to $\int_\Omega u\circ f d p_{\lambda}(\bm{\mu})$. This result contrasts with the existing literature that emphasizes conflicts between Commutativity and linear opinion pooling \citep{genest1984conflict} and that shows Commutativity justifies geometric opinion pooling \citep{molavi2018theory, dietrich2021fully}. 

Second, suppose the DM perceives ambiguity about the credibility of the experts, and her aggregated preference follows the Uncertainty Aversion axiom in \cite{gilboa1989maxmin}. In this case, Weak Commutativity characterizes a generalization of linear opinion pooling, where the DM has a set $\Lambda$ of weights for the experts, and the utility evaluation of an act $f$ is equal to $\min_{\lambda\in\Lambda}\int_\Omega u\circ f dp_{\lambda}(\bm{\mu}).$

In general, Weak Commutativity characterizes a novel belief aggregation model, \textit{dual-self aggregation}, whose interpretation is reminiscent of \cite{chandrasekher2022dual}. The DM has a collection $\bm{\Lambda}$ of sets of weights assigned to the experts. The utility evaluation of act $f$ is given by
\begin{align}\label{intromain}
\max_{\Lambda\in \bm{\Lambda}}\min_{\lambda\in \Lambda} \int_{\Omega} u\circ f dp_{\lambda}(\bm{\mu}).
\end{align}
The representation (\ref{intromain}) can be understood as a model where the DM optimally determines how much to learn from each expert $1,2,...,n$. The optimal weight $\lambda^*$ that solves (\ref{intromain}) is determined in an intrapersonal game played by Optimism and Pessimism. Optimism desires to learn more from the optimistic experts, which generates belief-based felicity; Pessimism desires to learn more from the pessimistic experts, which leads to practical and prudent decision-making. Formally, for every act $f$, Optimism chooses a set of weights $\Lambda\in\bm{\Lambda}$ with the goal of maximizing the expected utility to $f$, and then Pessimism chooses a weight $\lambda$ from $\Lambda$ with the goal of minimizing the expected utility to $f$. Section \ref{sec324} discusses how (\ref{intromain}) accommodates behavioral patterns that are unexplained by linear or multiple weight aggregation.

My second contribution lies in a series of propositions that link various axioms on the (non-)commutativity of belief aggregation to classic axioms on decision-making under uncertainty. The statements below all presume Pareto and Monotonicity. First, notice that (\ref{intromain}) can be written as dual-self expected utility \citep{chandrasekher2022dual}, which satisfies the C-Independence axiom \citep{gilboa1989maxmin}. It can be proved that Weak Commutativity is satisfied if and only if the DM's preference formed by aggregating expert preferences satisfies C-Independence. This result characterizes the family of invariant biseparable preferences \citep{ghirardato2004differentiating} in the setting of belief aggregation.

Second, recall from the discussion above that, in update-then-aggregate, the DM only learns expert posteriors and direct information about expert priors is missing. In Section \ref{sec51}, I introduce Thought Experiment \ref{exp2} that illustrates how this ``missing information'' problem leads the DM to perceive uncertainty about the reason for the \textit{ex-post} disagreement among the experts. If the DM is averse to this form of uncertainty, then she should be more pessimistic in the update-then-aggregate approach. To capture this form of uncertainty aversion, I introduce \textit{Pessimism to Update-then-Aggregate}, a non-commutativity condition stating that update-then-aggregate yields a more pessimistic outcome -- in terms of assigning (weakly) lower certainty equivalents to each act -- compared to aggregate-then-update. I show that {Pessimism to Update-then-Aggregate} implies the Uncertainty Aversion axiom in \cite{gilboa1989maxmin} if the size of the state space is weakly greater than 4.

Third, I introduce \textit{Moderate Commutativity}, which is a joint strengthening of Weak Commutativity and Pessimism to Update-then-Aggregate. The axiom captures a DM who has a neutral attitude to the uncertainty brought by the ``missing information'' problem in the update-then-aggregate. I establish the equivalence between Moderate Commutativity and the Independence axiom in \cite{anscombe1963definition}. Finally, I show that the strongest form of Commutativity leads to an impossibility result. Any aggregation rule that satisfies Pareto, Monotonicity, and Commutativity must be a \textit{dictatorship}, where the DM's subjective belief always coincides with the suggestion provided by one of the experts.


\subsection{Literature Review}
\label{litreview}

\noindent \textsc{Commutativity and Opinion Pooling} The Commutativity assumption was first introduced by \cite{madansky1964externally} (under the name \textit{External Bayesianity}) and was widely discussed by the literature on opinion pooling in statistics \citep{mcconway1978combination, weerahandi1978pooling, bordley1982multiplicative, genest1984characterization, genest1984conflict}. In the economics literature, in \cite{molavi2018theory}, Commutativity was named \textit{IIA} and is the central assumption characterizing the geometric pooling rule, where for beliefs $\mu_1,...,\mu_n$ provided by the experts, the result of belief aggregation is $\mu(\omega)=\Pi_{i=1}^n{\mu_i(\omega)^{\alpha_i}}/C$ for every $\omega\in \Omega$, where $\{\alpha_i\}_{i=1}^n$ are weights that sum up to one and $C$ is a normalization constant. Named \textit{Dynamic Rationality} by \cite{dietrich2021fully}, Commutativity was again key to the characterization of a collective choice model with geometric opinion pooling. Two variants of Commutativity deliver impossibility results for the widely-used linear opinion pooling model \citep{genest1984conflict} and the problem of selecting beliefs from a set of possible priors \citep{chambers2010bayesian}. This paper emphasizes that Commutativity is \textit{not} an appealing assumption for the aggregation of heterogeneous priors and proposes Weak Commutativity as a reasonable relaxation, which turns out to be consistent with linear opinion pooling.

\noindent \textsc{Aggregation and Ambiguity} This paper contributes to the literature on belief aggregation under uncertainty. \cite{cres2011aggregation} considered the problem of aggregating MEU experts with different sets of beliefs. Relatedly, \cite{gajdos2013decisions} and \cite{hill2011unanimity} studied the aggregation of sets of beliefs in a multiple-profile setting. \cite{nascimento2012ex} allows the experts to have general perceptions and attitudes to ambiguity. \cite{alon2016utilitarian} considered a society of SEU members with heterogeneous states and characterized a model where the society's preference has an MEU representation. \cite{stanca2021smooth} characterized an additively separable aggregation rule that resembles the smooth ambiguity model \citep{klibanoff2005smooth}. All these models assume the DM exhibits a certain form of aversion to the uncertainty brought by expert disagreements,\footnote{Specifically, these assumptions for the aversion to the uncertainty brought by expert disagreements include Expert Uncertainty Aversion (EUA) in \cite{cres2011aggregation}, Disagreement Aversion in \cite{gajdos2008attitude}, Convexity in \cite{nascimento2012ex}, and Ambiguity Aversion in \cite{stanca2021smooth}. These assumptions guarantee that the shape of the aggregator of expert evaluations must be concave. \cite{hill2011unanimity} did not impose a similar assumption on the aggregator but assumed the aggregated preference to be ambiguity averse.} which implies a negative attitude towards ambiguity (i.e. a preference for hedging). In this paper, the key assumption (\text{Weak Commutativity}) characterizes a dual-self aggregation model that accommodates a range of intermediate ambiguity attitudes. Section \ref{sec324} shows how this model accommodates patterns of ambiguity attitudes that cannot be explained by existing models. 

Another related paper is \cite{amarante2021aggregation}, which also allows for intermediate ambiguity attitudes in belief aggregation. In \cite{amarante2021aggregation}, the DM calculates the average of expert opinions using a capacity instead of a probability measure (or, weight) over the set of experts. Besides the difference in the representation, the current paper focuses on a different set of axioms on Commutativity and establishes a series of propositions that link them to classic axioms on decision-making under uncertainty.

The rest of the paper is organized as follows. Section \ref{sec2} presents the setup of my model. Section \ref{sec3} introduces Weak Commutativity and its role in the characterization of various aggregation models. Section \ref{sec5} links various aggregation axioms to classic axioms on decision-making under uncertainty. Section \ref{sec6} discusses the unanimity assumptions employed by this paper. 

%% file: Sections/figure1.tikz
\tikzset{every picture/.style={line width=0.75pt}} 

\begin{tikzpicture}[x=0.75pt,y=0.75pt,yscale=-1.05,xscale=1.05]

\shade[left color={rgb, 255:red, 230; green, 230; blue, 230}, right color=white] (132,100)  rectangle  (581.2,172);

\shade[left color={rgb, 255:red, 189; green, 189; blue, 189}, right color=white]  (351.1,115)  rectangle  (581.2,172);

\draw    (77,172) -- (578.2,172) ;
\draw [shift={(580.2,171.6)}, rotate = 179.95] [color={rgb, 255:red, 0; green, 0; blue, 0 }  ][line width=0.75]    (10.93,-3.29) .. controls (6.95,-1.4) and (3.31,-0.3) .. (0,0) .. controls (3.31,0.3) and (6.95,1.4) .. (10.93,3.29)   ;
\draw    (132,161.53) -- (132,172) ;
\draw    (281.1,161.33) -- (281.1,171.8) ;
\draw    (422.1,161.33) -- (422.1,171.8) ;

\draw (566,178.2) node [anchor=north west][inner sep=0.75pt]  [font=\small]  {$Time$};
\draw (83,145) node [anchor=north west][inner sep=0.75pt]  [font=\small] [align=left] {Expert Priors};
\draw (232,145) node [anchor=north west][inner sep=0.75pt]  [font=\small] [align=left] {Event $\displaystyle E$ realizes};

\draw (360,145) node [anchor=north west][inner sep=0.75pt]  [font=\small] [align=center] {Expert Posteriors};

\draw (114,177.2) node [anchor=north west][inner sep=0.75pt]  [font=\small] [align=left] {$\displaystyle t=0$};
\draw (262,178.2) node [anchor=north west][inner sep=0.75pt]  [font=\small] [align=left] {$\displaystyle t=1$};
\draw (403,179.2) node [anchor=north west][inner sep=0.75pt]  [font=\small] [align=center] {$\displaystyle t=2$};

\draw (152,120) node [anchor=north west][inner sep=0.75pt]   [align=center] {\textit{{\fontfamily{ptm}\selectfont Aggregate-then-Update}}};
\draw (360,120) node [anchor=north west][inner sep=0.75pt]   [align=center] {\textit{{\fontfamily{ptm}\selectfont Update-then-Aggregate}}};

\end{tikzpicture}

%% file: Sections/2.tex
\section{Setup}
\label{sec2}

I adopt the \cite{anscombe1963definition} setup. Let $X$ be a subset of a linear space (for example, a set of simple lotteries). Let $\Omega$ with $|\Omega|\geq 3$ be a finite state space and let $\Delta\Omega$ denote the distributions over $\Omega$. Let $\mathcal{E}=2^\Omega-\{\emptyset\}$ be the collection of all non-empty events, with generic element $E$. The set of acts $\mathcal{F}$ is the set of all functions $f:\Omega\rightarrow X$. In the usual abuse of notation, $x\in X$ denotes the constant act giving payoff $x$ in every state. Fix a binary relation $\succsim_0$ over $X$ that is represented by a linear and \textit{unbounded} $u(\cdot): X\rightarrow \mathbb{R}$.\footnote{The unboundedness of $u(\cdot)$ is equivalent to the following assumption on $\succsim_0$: for each $x,y\in X$ such that $x\succ y$, there exists $z,z'\in X$ such that
$
\frac{1}{2}z+\frac{1}{2}y\succsim_0 x\succ y\succsim_0 \frac{1}{2}x+\frac{1}{2}z'$.} The binary relation $\succsim_0$ is interpreted as the taste of the DM. A binary relation $\succsim$ over $\mathcal{F}$ is \textit{consistent with} $\succsim_0$, if for every $x,y\in X$, $x\succsim y$ if and only if $x\succsim_0 y$. Denote $\mathcal{R}_0$ as the set of all the binary relations over $\mathcal{F}$ that are consistent with $\succsim_0$.

\begin{definition}
A binary relation $\succsim$ over $\mathcal{F}$ is a continuous and monotone weak order if for every $f,g\in\mathcal{F}$ , the following conditions hold:
\begin{enumerate}
    \item \textit{Weak Order}: $\succsim$ is complete and transitive;
    \item \textit{Dominance}: if
    $
    f(\omega)\succsim g(\omega) \text{ for every } \omega\in\Omega
    $, then $f\succsim g$;
    \item \textit{Continuity}: the sets  $\{f'\in\mathcal{F}\mid f'\succsim f\}$ and $\{f'\in \mathcal{F}\mid f\succsim f'\}$ are closed.
\end{enumerate}
Let $\mathcal{R}$ be the collection of all continuous and monotone weak orders in $\mathcal{R}_0$.
\end{definition}

The DM forms her preference $\succsim \in\mathcal{R}$ by aggregating the information from the beliefs suggested by a set of experts. Formally, there is a finite set $N=\{1,..., i, ..., n\}$ experts. Each expert $i$ has a probability distribution $\mu_i\in \Delta\Omega$, which can be interpreted as a \textit{suggestion} that the DM should make the choices that maximize the SEU function $\int_{\Omega} u\circ f d\mu_i$, where $u:X\rightarrow \mathbb{R}$ is the DM's utility index. I call $\mu_i$ the \textit{suggestion} from expert $i$. A \textit{suggestion profile} is an element $\bm{\mu}=(\mu_1,...,\mu_n)\in (\Delta\Omega)^n$. At the center of my analysis is an \textit{aggregation rule} $\sigma$ that assigns a continuous and monotone preference to each suggestion profile. 

\begin{definition}\label{ruledef}
An aggregation rule is a function $\sigma:(\Delta\Omega)^n\rightarrow \mathcal{R}$. For every $\bm{\mu}\in (\Delta\Omega)^n$,
\begin{enumerate}
    \item $\succsim_{\bm{\mu}}$ denotes the binary relation $\sigma(\bm{\mu})$;
    \item $U_{\bm{\mu}}:\mathcal{F}\rightarrow \mathbb{R}$ is the representation for $\succsim_{\bm{\mu}}$ such that $U_{\bm{\mu}}(x)=u(x)$ for all $x\in X$. 
\end{enumerate}
\end{definition}

The aggregation rule $\sigma$ maps each suggestion profile $\bm{\mu}$ to a preference in $\mathcal{R}$. This multi-profile approach is necessary for the study of Commutativity, in which the suggestion profile $\bm{\mu}$ is changed when the experts update on the realized event. To understand the utility representation $U_{\bm{\mu}}:\mathcal{F}\rightarrow \mathbb{R}$, notice that for every act $f$, there exists a payoff $x^*\in X$ such that $f\sim_{\bm{\mu}}x^*$ by continuity. The representation $U_{\bm{\mu}}$ assigns utility $u(x^*)$ to the act $f$. It is easy to verify that $U_{\bm{\mu}}$ is uniquely defined and indeed represents $\succsim_{\bm{\mu}}$.

Next, I introduce a regularity condition that I impose on the aggregation rule $\sigma$ throughout this paper, which can be decomposed into two unanimity conditions stating the DM respects the consensus of the experts. For every $\mu\in\Delta\Omega$, I define the corresponding expected utility function by $EU_{\mu}(f)=\int_{\Omega} u\circ f d\mu$.

\begin{assumption} (Regularity) \label{assumption1}
For all suggestion profiles $\bm{\mu},\bm{\mu}'\in (\Delta\Omega)^n$, 
\begin{enumerate}
    \item (Pareto) if $EU_{\mu_i}(f)\geq EU_{\mu_i}(g)$ for every $i\in N$, then $U_{\bm{\mu}}(f)\geq U_{\bm{\mu}}(g)$;
    \item (Monotonicity) if $EU_{\mu'_i}(f)\geq EU_{\mu_i}(f)$ for every $i\in N$, then $U_{\bm{\mu'}}(f)\geq U_{\bm{\mu}}(f)$. 
\end{enumerate}
\end{assumption}

As mentioned above, for each expert $i$, the distribution $\mu_i$ can be interpreted as the suggestion for the DM to make choices that maximize the SEU function $EU_{\mu_i}(f)=\int_{\Omega}u\circ f d\mu_i$, where $u: X\rightarrow \mathbb{R}$ is the DM's utility index. Then, there is a one-to-one correspondence between the beliefs $\mu\in\Delta\Omega$ and the preferences represented by $EU_{\mu}(\cdot)$. Essentially, the aggregation rule can be viewed as a mapping from the profiles of expert expected utilities to the preference of the DM. In \text{Pareto}, $EU_{\mu_i}(f)\geq EU_{\mu_i}(g)$ for every $i\in N$ means all the experts suggest act $f$ should be chosen over $g$. In this case, \text{Pareto} states that the DM respects the consensus that $f$ is a better choice than $g$. While \text{Pareto} is a widely applied normative assumption for belief aggregation\footnote{More generally, the aggregation of preferences with identical tastes and heterogeneous beliefs.} \citep{cres2011aggregation, gajdos2008attitude, gajdos2013decisions, hill2011unanimity, nascimento2012ex, stanca2021smooth}, it also puts implicit restrictions on the aggregation model \citep{bommier2021disagreement}. Section \ref{sec6} discusses \text{Pareto} in more detail.

In the second part of Assumption \ref{assumption1} (\text{Monotonicity}), I compare the suggestion profile $\bm{\mu}$ to $\bm{\mu}'$. Compared to $\bm{\mu}$, all the experts attach a higher utility evaluation to act $f$ according to $\bm{\mu}'$, that is, all the experts are more optimistic about the value of act $f$.\footnote{\text{Monotonicity} is not a ``cardinal'' property, that is, it can be written in terms of the preferences represented by the SEU functions $EU_{\mu_i}$ and $EU_{\mu_i'}$. Denote the preferences as $\succsim$ and $\succsim'$ respectively and notice that $EU_{\mu'}(f)\geq EU_{\mu}(f)$ if and only if there exists an $x\in X$, such that $f\succsim' x$ and $x\succsim f$.} In this case, the DM also becomes more optimistic about the value of $f$. The key difference between \text{Pareto} and \text{Monotonicity} is that the latter is a comparative statics assumption on the DM's behavior in response to exogenous changes in the suggestions. This condition is analogous to the Weak Setwise Function assumption in \cite{mcconway1981marginalization}, the Weak Independence assumption in \cite{hill2011unanimity}, the Dominance assumption in \cite{gajdos2013decisions}, and a relaxation of the Aggregation of Expert Opinion (AEO) assumption in \cite{amarante2021aggregation}.\footnote{The interpretation of \text{Monotonicity} is related to the axiom of "positive association of social and individual values" in \cite{arrow1950difficulty}, which asserts ``the social welfare function is such that the social ordering responds positively to alterations in individual values or at least not negatively''.} I discuss the potential restrictions imposed by \text{Monotonicity} in Section \ref{sec6}.

For every event $E,F\in\mathcal{E}$ and distribution $\mu\in\Delta\Omega$ such that $\mu(E)>0$, define $\mu(F\mid E)=\mu(E\cap F)/\mu(E)$. The conditional distribution $\mu^E$ is defined such that $\mu^E(G)=\mu(G\mid E)$ for every $G\in\mathcal{E}$. For $\bm{\mu}\in(\Delta\Omega)^{n}$, I define $\mathcal{E}(\bm{\mu})=\{E\in\mathcal{E}\mid \mu_i(E)>0 \text{ for every }i=1,2,...,n\}$, which is the set of events on which the conditional distributions are well-defined. For $E\in\mathcal{E}(\bm{\mu})$, I define $\bm{\mu}^E=(\mu_1^E,...,\mu_n^E)$ as the profile of conditional distributions. 

For acts $f,g\in\mathcal{F}$, the act $fEg\in\mathcal{F}$ is defined by $[fEg]({\omega})=f(\omega)$ if $\omega\in E$ and $[fEg](\omega)=g(\omega)$ if $\omega\notin E$. Next, I define the update of the DM's preference $\succsim_{\bm{\mu}}\in\mathcal{R}$ following the standard notion in \cite{ghirardato2002revisiting}. 

\begin{definition} 
\label{condpref}
For $\bm{\mu}\in(\Delta\Omega)^n$ and $\succsim_{\bm{\mu}}\in\mathcal{R}$, an act $f$ is \textit{preferred over} $g$ \textit{conditional on} $E\in \mathcal{E}(\bm{\mu})$, denoted by $f\succsim_{\bm{\mu}}^E g$, if $
fEh\succsim_{\bm{\mu}} gEh \text{ for every } h\in \mathcal{F}.
$
\end{definition}

Definition \ref{condpref} is a dynamic consistency criterion for updating preferences.\footnote{To see it more concretely, notice that $f\succsim^E_{\bm{\mu}} g \text{ and } f\succsim^{E^c}_{\bm{\mu}} g\implies  f\succsim_{\bm{\mu}} g$, the dynamic consistency requirement over the information partition $\{E,E^c\}$ in \cite{sarin1998dynamic}, \cite{epstein2003recursive} and \cite{ellis2018dynamic}.} The definition provides foundations for Bayesian updating if the unconditional preference admits an SEU representation \citep{ghirardato2002revisiting}. However, it is well-known in the updating-under-ambiguity literature that Definition \ref{condpref} can be ``too strong'' for updating generic preference relations, as the conditional preference can be \textit{incomplete} due to the violation of the sure-thing principle. As will be clear in the next section, this problem is avoided in our main assumption (Axiom \ref{weak}). This is because Axiom \ref{weak} only puts a restriction on conditional preferences in situations where the set of expert suggestions $\{\mu_1,...,\mu_n\}$ satisfies \textit{rectangularity} with respect to the partition $\{E,E^c\}$ \citep{epstein2003recursive}. For MEU \citep{epstein2003recursive} and more general classes of preferences \citep{chandrasekher2022dual}, rectangularity is known to guarantee well-defined conditional preferences generated by Definition \ref{condpref}.

Sections \ref{sec3} and \ref{sec5} provide formal definitions of the relevant axioms and discuss their implications in belief aggregation.

%% file: Sections/3.tex
\section{Weakly Commutative Aggregation Rules}
\label{sec3}

This Section introduces \text{Weak Commutativity} and discusses its implications. Section \ref{sec31} formally defines the assumption and discusses how it circumvents the ``missing information'' criticism raised by Thought Experiment \ref{exp1}. Section \ref{sec32} shows that Weak Commutativity characterizes various rules for aggregating heterogeneous priors, including linear aggregation, multiple-weight aggregation, and dual-self aggregation.

\subsection{Defining Weak Commutativity} 
\label{sec31}

This section defines the key assumption, Weak Commutativity, which requires that the updating and aggregation operations commute \textit{only} when the experts agree on the probability of every state outside of the realized event $E$. Definition \ref{def41} below formally states this restriction on the suggestion profile $\bm{\mu}\in (\Delta\Omega)^n$.

\begin{definition}\label{def41}
    For $\bm{\mu}\in(\Delta\Omega)^n$ and  and $E\in\mathcal{E}(\bm{\mu})$, expert disagreement is restricted within $E$ if 
    \begin{align} \label{restrict}
    \mu_1(\omega)=...=\mu_n(\omega) \text{ for every } \omega\notin E.
    \end{align}
\end{definition}

There are two reasons why I introduce this concept. First, if expert disagreement is restricted within $E$, then the commutativity between the updating and aggregation operations is free from the critiques raised in Thought Experiment \ref{exp1}. Recall that in Thought Experiment \ref{exp1}, Alice and Bob attach different probabilities to the realized event $E$ \textit{ex-ante}. In aggregate-then-update, the DM is aware of this and adjusts the credibility of the experts accordingly, while in update-then-aggregate, this information is missing. However, if expert disagreement is restricted within $E$, then all the experts attach identical priors to this event, that is, 
$
\mu_1(E)=\mu_2(E)=...=\mu_n(E).
$ Therefore, in aggregate-then-update, even though the DM knows $\mu_1(E),..., \mu_n(E)$, no further inference about expert credibility can be made from it. 

Second, if expert disagreement is restricted within $E$, then if the DM takes an arbitrary subset $P^*$ of the convex hull of the expert suggestions $Conv(\{\mu_1,...,\mu_n\})$ as her set of beliefs, $P^*$ must satisfy \textit{rectangularity} \citep{epstein2003recursive}. Rectangularity is a type of stochastic independence condition that helps to avoid the negative result associated with the dynamic consistency requirement in Definition \ref{condpref}. For example, when $P^*$ is rectangular, if the DM has MEU preference $\succsim$ with $P^*$ as the set of beliefs, then the conditional preference $\succsim^E$ is the MEU preference with the prior-by-prior updating of $P^*$ as the set of beliefs.

To show that $P^*$ must satisfy rectangularity, I first formally define this property. A set of distributions $P$ is \textit{rectangular} with respect to $\{E,E^c\}$ if for every $p^1,p^2,p^3\in P$ there exists a $p^4\in P$ such that for every $F\in\mathcal{E}$,
\begin{align}\label{rectangularity}
p^4(F)=p^3(E)p^1(F\mid E)+p^3(E^c)p^2(F\mid E^c).
\end{align}
Take arbitrary distributions $p^1,p^2,p^3$ from $P^*$, and take $F\in\mathcal{E}$. Since disagreement is restricted within $E$, the experts must agree that $Prob(E)=\alpha$ for some $\alpha\in (0,1)$. Moreover, the experts agree on the probability of every state in $E^c$, so they must agree that $Prob(F\mid E^c)=\beta$ for some $\beta\in [0,1]$. Since $P^*\subset Conv(\{\mu_1,...,\mu_n\})$,  $p^k(E)=\alpha$ and $p^k(F\mid E^c)=\beta$ for every $k=1,2,3,4$. Therefore, (\ref{rectangularity}) states that $p^4(F)=\alpha p^1(F\mid E)+(1-\alpha)\beta$, which is satisfied if we set $p^4=p^1$. As a result, rectangularity is satisfied.

Next, I formally state the axiom of Weak Commutativity.

\begin{axiom} (Weak Commutativity)    \label{weak}
    For $\bm{\mu}\in(\Delta\Omega)^n$ and $E\in\mathcal{E}(\bm{\mu})$, if expert disagreement is restricted within $E$, then 
\begin{align}\label{a1res}
f\succsim_{\bm{\mu}}^E g\iff f\succsim_{\bm{\mu}^E} g.
\end{align}
\end{axiom}

Weak Commutativity only requires that updating and aggregation \textit{commute}  when expert disagreement is restricted within the realized event $E$. The reason for introducing Weak Commutativity is twofold. First, as discussed above, this axiom relaxes Commutativity in a way that circumvents the concerns raised in Thought Experiment \ref{exp1}. Second, Weak Commutativity is a normatively appealing criterion for aggregating heterogeneous priors. Suppose expert disagreement is restricted within $E$. In that case, if the DM aggregates the expert priors first, then she should only pay attention to event $E$, because all the experts completely agree outside of $E$. On the other hand, if the experts update on $E$ first, the DM should still only focus on $E$ when aggregating expert posteriors, because it is impossible for $E^c$ to realize. As a result, aggregate-then-update should produce the same outcome as update-then-aggregate.




The next section shows that Weak Commutativity characterizes various rules for aggregating heterogeneous priors.

\subsection{Weak Commutativity and Aggregation Rules}
\label{sec32}

This section characterizes the aggregation rules that satisfy Weak Commutativity. For expositional purposes, I consider a series of special cases that build up to the general case, while I employ the opposite order in the proofs in Appendix \ref{seca} and \ref{secb}. Sections \ref{sec321} and \ref{sec322} discuss the implications of Weak Commutativity in ambiguity-free and ambiguity-averse settings respectively. Section \ref{sec323} shows that, in general, Weak Commutativity characterizes a dual-self belief aggregation model, where the DM optimally chooses the weight $\lambda^*\in \Delta N$ that determines how much to learn from each expert.
\subsubsection{Linear Aggregation Rules} \label{sec321}

This section presents a representation theorem that relates Weak Commutativity to the classic model of linear belief aggregation \citep{stone1961opinion, degroot1974reaching}. I first define linear aggregation rules. Let $\Delta N$ denote the distributions over $N$. A \text{credibility weight} (\textit{weight}) is a distribution $\lambda=(\lambda_i)_{i=1}^n\in\Delta N$, which reflects the confidence assigned to each expert. For each suggestion profile $\bm{\mu}$, the DM calculates the \text{linearly aggregated belief} with respect to weight $\lambda$, which is defined by $p_{\lambda}(\bm{\mu})=\sum_{i=1}^n \lambda_i \mu_i$. Formally, $\sigma: (\Delta\Omega)^n\rightarrow \mathcal{R}$ is a \textit{linear aggregation rule}, if there exists a weight $\lambda\in\Delta N$, such that for every $\bm{\mu}\in(\Delta\Omega)^n$, 
\begin{align}\label{linearagg}
U_{\bm{\mu}}(f)=\int_\Omega u\circ f d p_{\lambda}(\bm{\mu}).
\end{align}
Since \cite{degroot1974reaching}, the linear aggregation rule has been widely applied in the literature on non-Bayesian social learning. Theorem \ref{ambfree} identifies the role played by Weak Commutativity in the characterization of this family of aggregation rules. First, a preference relation $\succsim$ satisfies assumption P2 (also known as \textit{sure-thing principle}), if for every acts  $f,g,h\in\mathcal{F}$, for every $E'\subset \Omega$, it holds that $fE'g\succsim g \iff fE'h\succsim gE'h$. 

\begin{theorem}\label{ambfree}
For a regular aggregation rule $\sigma$ such that $\succsim_{\bm{\mu}}$ satisfies P2 for every $\bm{\mu}\in(\Delta\Omega)^n$, Weak Commutativity holds if and only $\sigma$ is a linear aggregation rule. 
\end{theorem}

Theorem \ref{ambfree} links the weakening of Commutativity proposed in Section \ref{sec31} to the widely known linear aggregation rule. 
The relationship between Commutativity and linear aggregation has been discussed by the literature. \cite{genest1984conflict} proved the inconsistency between Commutativity and linear aggregation: a combination of both conditions leads to dictatorial belief aggregation, where the DM ignores the opinions of all but one expert. Since linear aggregation violates Commutativity, the Bayesian update of the aggregation of expert priors [\textit{i.e.} the update of $p_{\lambda}(\bm{\mu})$] is different from the aggregation of expert posteriors [$p_{\lambda}(\bm{\mu}^E)$]. This creates dynamic inconsistency if the DM takes $p_{\lambda}(\bm{\mu})$ as her \textit{ex-ante} belief and $p_{\lambda}(\bm{\mu}^E)$ as her \textit{ex-post} belief, which has been viewed as raising questions about the
rationality of linear opinion aggregation. For example, \cite{sep-epistemology-social} stated that a DM who employs linear opinion aggregation can ``be Dutch booked, meaning they will accept a series of bets guaranteed to lose them money.'' In the economics literature, \cite{dietrich2021fully} complained that linear belief aggregation creates a DM that is ``statically
rational, but dynamically irrational''. 

However, as Thought Experiment \ref{exp1} illustrates, Commutativity may not be a normatively appealing criterion for the aggregation of heterogeneous priors. Different from \cite{dietrich2021fully} and \cite{sep-epistemology-social}, this paper argues that Commutativity instead of linear opinion aggregation should be reformulated to avoid the impossibility result in \cite{genest1984conflict}. Theorem \ref{ambfree} shows that linear belief aggregation is consistent with Weak Commutativity, and is the only regular aggregation rule that satisfies both Weak Commutativity and Savage's P2.

\subsubsection{Multiple Weight Aggregation Rules}\label{sec322}

 In many situations, the DM faces ambiguity in determining which expert possesses greater expertise or should be more trusted. For a DM who is averse to such ambiguity, Weak Commutativity characterizes a multiple-weight aggregation rule that is reminiscent of \cite{cres2011aggregation} and \cite{gajdos2013decisions}. Formally, an aggregation rule is a \textit{multiple-weight aggregation rule} (or multiple-weight rule), if there exists a closed and convex set of weights $\Lambda\subset \Delta N$ such that for every $\bm{\mu}\in(\Delta\Omega)^n$, 
\begin{align}\label{multiagg}
    U_{\bm{\mu}}(f)=\min_{\lambda\in\Lambda} \int_{\Omega} u\circ f dp_{\lambda}(\bm{\mu}). 
\end{align}
Recall that $p_{\lambda}(\bm{\mu})=\sum_{i=1}^n \lambda_i \mu_i$ is the linear aggregation of expert beliefs with weight $\lambda$. In (\ref{multiagg}), the set of weights $\Lambda$ allows the DM to assign different weights to the opinions suggested by each expert and leaves room for ambiguity regarding the credibility of experts. Theorem \ref{ambaverse} below shows that Weak Commutativity characterizes this family of aggregation rules when the DM is ambiguity-averse in the sense of \cite{gilboa1989maxmin}. Formally, a preference $\succsim$ is {ambiguity-averse}, if for every $f,g\in\mathcal{F}$ and $\alpha\in (0,1)$, if $f\sim g$, then $\alpha f+(1-\alpha)g\succsim g$.

\begin{theorem}\label{ambaverse}
For a regular aggregation rule $\sigma$ such that $\succsim_{\bm{\mu}}$ is ambiguity-averse for every $\bm{\mu}\in(\Delta\Omega)^n$, Weak Commutativity holds if and only if $\sigma$ is a multiple-weight aggregation rule.
\end{theorem}

\subsubsection{Dual-self Aggregation Rules}\label{sec323}

This section introduces dual-self aggregation rule, which is characterized by Weak Commutativity. In the previous section, it is assumed that the DM's preference always satisfies ambiguity aversion. However, individuals may exhibit different attitudes to ambiguity in different circumstances \citep{trautmann2015ambiguity}. In general, Weak Commutativity characterizes a family of aggregation rules that allows the DM to exhibit different attitudes to ambiguity. Similar to \cite{chandrasekher2022dual}, the DM is characterized by a sequential game played by ``dual selves,'' Optimism, and Pessimism. I first introduce the representation of the dual-self aggregation. 

Recall that a {weight} is a distribution $\lambda\in\Delta N$. Let $\mathcal{K}$ denote the set of all nonempty closed, convex sets of weights that belong to $\Delta N$, endowed with the Hausdorff topology. A \textit{weight-set collection} is a non-empty compact collection $\bm{\Lambda}\subset\mathcal{K}$, that is, each element $\Lambda\in\bm{\Lambda}$ is a nonempty closed, convex set of weights. 

\begin{definition} (Dual-self Aggregation)\label{defdual}
    An aggregation rule is a dual-self aggregation rule (or dual-self rule), if there exists a weight-set collection $\bm{\Lambda}$ such that, for every $\bm{\mu}\in(\Delta\Omega)^n$,
\begin{equation}\label{dual}
U_{\bm{\mu}}(f)=\max _{\Lambda\in\bm{\Lambda}} \min _{\lambda\in \Lambda} \int_{\Omega} u\circ fd p_{\lambda}(\bm{\mu}).
\end{equation}
\end{definition}

The representation (\ref{dual}) can be understood as a model where the DM optimally determines how much to learn from each expert $1,2,...,n$. Similar to the interpretation of \cite{chandrasekher2022dual}, the optimal weight $\lambda^*$ that solves (\ref{dual}) can be interpreted as being determined in an intrapersonal game played by \textit{Optimism} and \textit{Pessimism}. When determining the evaluation of the act $f$, Optimism desires the DM to learn more from the experts who are optimistic about $f$, which generates belief-based felicity; on the other hand, Pessimism desires the DM to learn from those being pessimistic about $f$, which leads to practical and prudent decision-making. Formally, Optimism first chooses a set of weights $\Lambda$ with the goal of maximizing the expected utility of act $f$; second, Pessimism chooses the weight $\lambda^*\in\Lambda$ with the goal of minimizing the expected utility of $f$. The aggregated belief is calculated according to $p_{\lambda^*}(\bm{\mu})=\sum_{i=1}^n \lambda^*_i \mu_i$.

As two special cases, (\ref{dual}) is a linear aggregation rule if $\bm{\Lambda}$ and $\Lambda$ are both singletons and a multiple-weight rule if the belief-set collection $\bm{\Lambda}$ is a singleton. The dual-self aggregation rule accommodates other behavioral patterns that cannot be explained by linear or multiple-weight aggregation, some of which are discussed in Section \ref{sec324}. Theorem \ref{thmdual} shows that Weak Commutativity characterizes dual-self aggregation rules.

\begin{theorem}
\label{thmdual}
For a regular aggregation rule $\sigma$, Weak Commutativity holds if and only if $\sigma$ is a dual-self aggregation rule. 
\end{theorem}

\textsc{Sketch of the Proof} The proof of Theorem \ref{thmdual} can be found in Appendix \ref{seca}. Here I sketch the proof for the ``only if'' part of the theorem. The regularity of the aggregation rule guarantees the existence of an \textit{aggregation functional} $I:\mathbb{R}^N\rightarrow \mathbb{R}$ such that for every $\bm{\mu}\in(\Delta\Omega)^n$ and $f\in\mathcal{F}$, $U_{\bm{\mu}}(f)=I(\bm{\mu}\cdot u_f)$. Here, $\bm{\mu}\cdot u_f=(EU_{\mu_1}(f),..., EU_{\mu_n}(f))\in\mathbb{R}^N$ is the \textit{evaluation profile} according to $\bm{\mu}$, which lists the expected utilities of $f$ evaluated according to each expert's suggestion in $\bm{\mu}$.

If the expert disagreement is restricted within event $E$ (as is required by Weak Commutativity), then all the experts agree that $Prob(E)=\alpha$ for some $\alpha$ between zero and one. The \text{martingale} property of Bayesian update then implies
$
\bm{\mu}=\alpha\bm{\mu}^E+(1-\alpha)\bm{\mu}^{E^c}.
$
By the linearity of inner products, the evaluation profile according to $\bm{\mu}$ is an $\alpha$-mixture between the evaluation profile according to $\bm{\mu}^E$ and $\bm{\mu}^{E^c}$: 
\begin{align}
\bm{\mu}\cdot u_f=\alpha \left(\bm{\mu}^E \cdot u_f\right)+(1-\alpha)\left(\bm{\mu}^{E^c}\cdot u_f\right). 
\end{align}
Notice that $\bm{\mu}\cdot u_f$, $\bm{\mu}^E \cdot u_f$ and $\bm{\mu}^{E^c}\cdot u_f$ are all evaluation profiles, which belong to the domain of the aggregation functional $I(\cdot):\mathbb{R}^N\rightarrow\mathbb{R}$. The key step of the proof is to prove by Weak Commutativity that the aggregation of the evaluation profile according to $\bm{\mu}$ is an $\alpha$-mixture of the aggregation of the counterparts according to $\bm{\mu}^E$ and $\bm{\mu}^{E^c}$:
\begin{align}
\label{linear1}
I(\bm{\mu}\cdot u_f)=\alpha I \left(\bm{\mu}^E \cdot u_f\right)+(1-\alpha)I\left(\bm{\mu}^{E^c}\cdot u_f\right). 
\end{align}
To complete the proof sketch, notice that if the expert disagreement is restricted within event $E$ (as is required by Weak Commutativity), then all the expert posteriors conditional on ${E^c}$ are identical, and the evaluation profile $\bm{\mu}^{E^c}\cdot u_f$ is constant over $N$. Since $\bm{\mu}$ and $f$ can be arbitrability chosen, (\ref{linear1}) implies the aggregation functional $I$ is a constant linear functional. An application of \cite{chandrasekher2022dual} then leads to the characterization of dual-self aggregation.

\subsubsection{Examples for Dual-self Aggregation} \label{sec324}

It is well-known that a DM may appear ambiguity-averse in one decision problem while being ambiguity-seeking in another; for a survey, see \cite{trautmann2015ambiguity}. Another fact known in the literature is that disagreement among experts leads the DM to perceive ambiguity \citep{cabantous2007ambiguity,baillon2012aggregating}. However, how different forms of disagreement may lead to different ambiguity attitudes of the DM remains less studied in both the theoretical and empirical literature. The dual-self aggregation rule provides a useful tool for studying the relationship between disagreement and ambiguity attitudes. To illustrate this point, I consider the two examples below. 

\textsc{Expert Credibility and Ambiguity Attitude}  Dual-self aggregation allows the DM's ambiguity attitude to be dependent on the \textit{credibility} of the experts who disagree with each other. For example, with $N=\{1,2,3,4\}$, consider the following special case of dual-self aggregation:
\begin{align} \label{identityamb}
    U_{\bm{\mu}}(f)=.8\min_{\lambda\in [.25,.75]} EU_{\lambda\mu_1+(1-\lambda)\mu_2}[f]+.2\max_{\lambda'\in [.25,.75]} EU_{\lambda'\mu_3+(1-\lambda')\mu_4}[f].
\end{align}

In this example, the DM divides the experts $N=\{1,2,3,4\}$ into two groups: $\{1,2\}$ and $\{3,4\}$. She attaches a total weight of $.8$ to experts 1 and 2 and a total weight of $.2$ to experts 3 and 4. That is, the DM believes that experts 1 and 2 are much more reliable than experts 3 and 4. If she forms her preference according to (\ref{identityamb}), then her ambiguity attitude is dependent on whether the experts in $\{1,2\}$ or $\{3,4\}$ disagree with each other. To illustrate the intuition behind this point, I assume that the state space is binary, $\Omega=\{H,L\}$, and consider the two cases below.

\textit{Case 1: Disagreement between 1 and 2} $\ $ Suppose that experts 1 and 2 disagree with each other, say, $\mu_1(H)=.2$ and $\mu_2(H)=.8$, while $\mu_3(H)=\mu_4(H)=.5$. In this case, the DM feels anxious because the two experts whose opinions are the most important to her significantly disagree with each other. This feeling of anxiety leads to ambiguity aversion. According to (\ref{identityamb}), the preference of the DM will be represented by the MEU function $U(f)=\min_{p\in [.38,.62]}pu[f(H)]+(1-p)u[f(L)]$.

\textit{Case 2: Disagreement between 3 and 4} $\ $ Suppose that experts 3 and 4 disagree with each other, who play less important roles in the formation of the DM's set of beliefs. Suppose $\mu_3(H)=.2$ and $\mu_4(H)=.8$ and $\mu_3(H)=\mu_4(H)=.5$. In this situation, it can be attractive to learn from the more optimistic expert in $\{3,4\}$ because, even if that expert is wrong, the DM's belief cannot be distorted too much since the total weight assigned to $\{3,4\}$ is only .2. Moreover, learning more from the more optimistic expert is enjoyable since it induces belief-related felicities. As a result, the DM exhibits ambiguity-seeking behaviors. According to (\ref{identityamb}), the preference of the DM will be represented by the maxmax expected utility $U(f)=\max_{p\in [.47,.53]}pu[f(H)]+(1-p)u[f(L)]$.

\textsc{Expert Optimism and Ambiguity Attitude} Dual-self aggregation also allows the DM's ambiguity attitude to be dependent on whether the disagreement is among the more optimistic or the more pessimistic experts. For example, with $N=\{1,2,3\}$, consider the following aggregation rule: 
 \begin{align}\label{median}
      U_{\bm{\mu}}(f)=\text{median}[EU_{\mu_1}(f),EU_{\mu_2}(f),EU_{\mu_3}(f)].
 \end{align}
The right-hand side of (\ref{median}) is the median of the profile of utility evaluations offered by the three experts. It can be verified that $\text{median}[\cdot,\cdot,\cdot]$ is a constant linear functional so (\ref{median}) belongs to the family of dual-self aggregation rules.\footnote{The utility function in (\ref{median}) can be rewritten in the dual-self form. Specifically, $U_{\bm{\mu}}(f)=\max\{\min_{\mu\in\{\mu_1,\mu_2\}}EU_{\mu}[f],\min_{\mu\in\{\mu_2,\mu_3\}}EU_{\mu}[f],\min_{\mu\in\{\mu_1,\mu_3\}}EU_{\mu}[f]\}$.} In the representation, the DM ignores the most optimistic expert due to the concern that the expert can be overoptimistic, and \textit{vice versa}. The following two cases show that the DM can exhibit different preferences for hedging in different circumstances. Similar to the previous section, I use $\bm{\mu}\cdot u_f$ to denote $[EU_{\mu_1}(f),EU_{\mu_2}(f), EU_{\mu_3}(f)]$, which is the profile of the expected utilities of $f$ evaluated according to each expert's suggestion in $\bm{\mu}$. 

\textit{Case 1: Optimism-led Disagreement} $\ $ Consider acts $f$ and $g$ such that 
$\bm{\mu}\cdot u_f=(0,0,2)$ and $\bm{\mu}\cdot u_g=(0,2,0)$, and it is easy to verify that $U_{\bm{\mu}}(f)=U_{\bm{\mu}}(g)=0$ according to (\ref{median}). For both acts, the expert who disagrees with the majority of $\{1,2,3\}$ is the most \textit{optimistic} one. If the DM is concerned that the most optimistic expert can be overoptimistic, then she should prefer hedging between $f$ and $g$, which can be interpreted as ambiguity aversion. Indeed, if we take the $\frac{1}{2}$-mixture of $f$ and $g$, then $\bm{\mu}\cdot u_{\frac{1}{2}f+\frac{1}{2}g}=(0,1,1)$, which implies $U_{\bm{\mu}}\left(\frac{1}{2}f+\frac{1}{2}g\right)=1>0$. 

\textit{Case 2: Pessimism-led Disagreement} $\ $ Symmetrically, consider acts $f'$ and $g'$ such that 
$\bm{\mu}\cdot u_{f'}=(0,0,-2)$ and $\bm{\mu}\cdot u_{g'}=(0,-2,0)$. For both acts, the expert who disagrees with the majority of $\{1,2,3\}$ is the most \textit{pessimistic} one. If the DM worries that the most pessimistic expert can be overpessimistic, then she should prefer not to hedge between $f'$ and $g'$, which can be interpreted as ambiguity seeking. In this case, take the $\frac{1}{2}$-mixture of $f'$ and $g'$ and it is immediate to verify that $\bm{\mu}\cdot u_{\frac{1}{2}f'+\frac{1}{2}g'}=(0,-1,-1)$, which implies $U_{\bm{\mu}}\left(\frac{1}{2}f'+\frac{1}{2}g'\right)=-1<U_{\bm{\mu}}(f')=0$.

%% file: Sections/5.tex
\section{(Non-)Commutativity and Decision-Making under Uncertainty}
\label{sec5}

This section introduces a series of propositions that link various aggregation axioms (including Weak Commutativity) to classic axioms on decision-making under uncertainty, including C-Independence, Uncertainty Aversion \citep{gilboa1989maxmin}, and Independence \citep{anscombe1963definition}. These results provide a new perspective for the foundation of various models in the setting of belief aggregation, including DSEU (dual-self expected utility, \citealt{chandrasekher2022dual}), MEU, and SEU. 

Recall from Theorem \ref{thmdual} that, for a regular aggregation rule $\sigma$, Weak Commutativity is equivalent to $\sigma$ being a dual-self aggregation rule. I start this section with a result closely related to Theorem \ref{thmdual}, which states the equivalence between Weak Commutativity and that the DM's preference satisfies the \text{C-Independence} axiom. A preference $\succsim$ satisfies \textit{C-Independence} if for every $f,g\in \mathcal{F}$, $x\in X$ and $\alpha\in (0,1)$,
$
f\succsim g\iff \alpha f+(1-\alpha)x\succsim \alpha g+(1-\alpha)x.
$

\begin{theorem} \label{thm41}
A regular aggregation rule satisfies Weak Commutativity if and only if $\succsim_{\bm{\mu}}$ satisfies C-Independence for every $\bm{\mu}\in(\Delta\Omega)^n$.
\end{theorem}

 C-Independence is the key axiom for invariant biseparable preferences \citep{ghirardato2004differentiating},\footnote{ Notable special cases of invariant biseparable preferences include $\alpha$-MEU and Choquet expected utility.} and Theorem \ref{thm41} characterizes this family of preferences in the setting of belief aggregation. The ``only if'' side of the theorem is a corollary of Theorem \ref{thmdual}: if an aggregation rule satisfies Weak Commutativity, then it is a dual-self aggregation rule in (\ref{dual}). This expression can be rewritten as a dual-self expected utility representation in \cite{chandrasekher2022dual}, which satisfies C-Independence. For the ``if'' side, we need to show C-Independence, a property for the DM's preference, implies the constant linearity of the aggregation functional $I$. Appendix \ref{secb} gives the detail of the proof.

 Theorem \ref{thm41} shows that in belief aggregation problems, if one finds Weak Commutativity appealing due to the reasons explained in Section \ref{sec31}, then the result of aggregating expert beliefs must satisfy C-Independence. This result can be viewed as providing a justification for C-Independence in the setting of belief aggregation, complementing the justification for the axiom in \cite{gilboa1989maxmin} that the mixture between $f$ and a constant $x$ is more easily visualized than the mixture between $f$ and an arbitrary act. 

\subsection{Non-Commutativity and Ambiguity Aversion}
\label{sec51}
In this section, I consider a DM whose preference is sensitive to the order of updating and aggregation. Recall that the update-then-aggregate approach is subject to the ``missing information'' problem since the expert priors are unobservable. In this section, I introduce Thought Experiment \ref{exp2} to illustrate how this missing information problem can give rise to uncertainty about the reason for the \textit{ex-post} disagreement among the experts. To capture a DM who is averse to this form of uncertainty, I introduce the non-commutativity condition, \textit{Pessimism to Update-then-Aggregate}, which states update-then-aggregate yields a more pessimistic outcome -- in terms of assigning (weakly) lower certainty equivalents to each act -- compared to aggregate-then-update. 
\begin{experiment}\label{exp2} 
 (Missing Information about Disagreement) Consider a state space $\Omega=\{hH,lH, hL, lL\}$, where $h,l$ (\textit{resp.} $H,L$) represents the rainfall levels in 2024 (\textit{resp.} 2025). The DM aggregates the beliefs of two experts, Alice and Bob, with priors $\mu^*_A=(.4,.1,.1,.4)$ and $\mu^*_B=(.1,.4,.4,.1)$ respectively. That is, the expert priors agree that the probability of $H$ is .5, however, conditional on $h$, their posteriors assign different probabilities to $H$, which are respectively $.8$ and $.2$. For $\mu_A^*$ and $\mu_B^*$, the expert posteriors are different because they interpret signal $h$ differently. In general, there are \textit{infinitely many} other profiles of joint priors $\{\mu_i'\}_{i=A,B}$ according to which the posterior belief for $H$ conditional on $h$ is $.8$ for Alice and $.2$ for Bob,\footnote{To guarantee that $[\mu_A'(H\mid h),\mu_B'(H\mid h)]=(.2,.8)$, it suffices to assume $\mu_A'(hH)=\mu_A'(hL)/4$ and $\mu_B'(hH)=4\mu_B'(hL)$.} each providing a \textit{reason} explaining why the experts arrive at different beliefs about $H$ \textit{ex-post}. 

Consider the DM's belief for $H$ (high rainfall for 2025) given that $h$ is realized for 2024. In the \textit{update-then-aggregate} approach, the DM only observes the experts' posteriors are .8 and .2. The DM has no idea about which joint priors generate these posteriors, in other words, she perceives uncertainty about the reason for the \textit{ex-post} disagreement among the experts. On the other hand, in the \textit{aggregate-then-update} approach, since the DM observes the experts' joint priors, the reason why the expert posteriors are different is clearly that the experts disagree about the interpretation of signal $h$. This resolves the uncertainty described above. 
\end{experiment}


Recall that Thought Experiment \ref{exp1} illustrates how Commutativity can be violated due to the missing information about expert credibility. Thought Experiment \ref{exp2} provides another reason why Commutativity can be violated: a DM who is averse to the uncertainty about the reason for \textit{ex-post} disagreement should be more pessimistic in update-then-aggregate than in the other case. Next, I formally define the pessimism to update-then-aggregate. First, a binary relation $\succsim$ is \textit{more pessimistic than} $\succsim'$ if for every $f\in\mathcal{F}$ and $x\in X$,
$
f\succsim x \implies f\succsim' x. 
$
That is, whenever $\succsim$ deems $f$ more valuable than $x$, the binary relation $\succsim'$ must reach the same conclusion.

\begin{axiom} (Pessimism to Update-then-Aggregate) \label{axiom4} 
For every $\bm{\mu}\in (\Delta\Omega)^n$ and $E\in\mathcal{E}(\bm{\mu})$ such that $\mu_1(E)=...=\mu_n(E)$, it holds that
$
\succsim_{\bm{\mu}^E} \text{ is more pessimistic than }\succsim_{\bm{\mu}}^E.
$
\end{axiom}

 Axiom \ref{axiom4} states the result of update-then-aggregate is more pessimistic than the result of aggregate-then-update, which is a non-commutativity assumption about the ``aversion'' to update-then-aggregate over the other case. This is because the DM is uncertain about the reason for the \textit{ex-post} disagreement among the experts in the former case. In the axiom, I restrict attention to the situtation where
\begin{align}\label{pessimismassumption}
    \mu_1(E)=...=\mu_n(E),
\end{align}
that is, all the experts agree on the probability of the $E$, so the DM does not infer about the credibility of the expert from the realization of event $E$. 

 \textsc{Remark.} Axiom \ref{axiom4} does not require that expert disagreement is restricted within $E$ (Definition \ref{def41}). As a result, the axiom puts a restriction on conditional preferences in situations where rectangularity can be violated. According to Definition \ref{condpref}, the DM's conditional preference $\succsim_{\bm{\mu}}^E$ can be \textit{incomplete}. Since $\succsim_{\bm{\mu}}^E$ reflects the strongest possible requirement for dynamic consistency, a weak order that models decision-making conditional on $E$ should be an \textit{extension} or \textit{completion} of the order $\succsim_{\bm{\mu}}^E$. In this case, the interpretation of Axiom \ref{axiom4} remains unchanged because it can be shown that $\succsim_{\bm{\mu}^E}$ is more pessimistic than any weak order in $\mathcal{R}$ that extends $\succsim_{\bm{\mu}}^E$. Finally, in the definition that $\succsim$ is more pessimistic than $\succsim'$, the two binary relations are not required to be complete, so the expression in Axiom \ref{axiom4} is well-defined.

\begin{proposition}
\label{prop41}
Suppose $|\Omega|\geq 4$. If a regular aggregation rule satisfies Pessimism to Update-then-Aggregate, then $\succsim_{\bm{\mu}}$ is ambiguity averse for every $\bm{\mu}\in(\Delta\Omega)^n$.
\end{proposition}

Proposition \ref{prop41} shows that Axiom \ref{axiom4}, which reflects the aversion to uncertainty regarding the reason for disagreement, implies the DM exhibits ambiguity aversion. The definition of ambiguity aversion is usually justified by the ``hedging'' argument that mixing two Anscombe-Aumann acts $f$ and $g$ hedges the uncertainty about payoffs across the states. Proposition \ref{prop41} provides a novel way to understand ambiguity aversion in the setting of belief aggregation. 

In Proposition \ref{prop41}, I assume that the cardinality of the state space is weakly greater than 4. This assumption is necessary because when $|\Omega|\leq 3$, Axiom \ref{axiom4} (Pessimism to Update-then-Aggregate) is weaker than Axiom \ref{weak} (Weak Commutativity). In that case,  Axiom \ref{axiom4} would have no implication in the DM's ambiguity attitude. To illustrate this point, first notice that when $|\Omega|\leq 2$, both axioms are innocuous because the experts' posterior beliefs are either degenerate (when the realized event $E$ is a singleton) or equal to the prior beliefs (when $E=\Omega$).  Now suppose $|\Omega|=3$, $E=\{1,2\}$ and $E^c=\{3\}$. In this case, Pessimism to Update-then-Aggregate requires $\mu_1(E)=...=\mu_n(E)$, which implies $\mu_1(\omega=3)=...=\mu_n(\omega=3)$, that is, expert disagreement is restricted within $E$ (Definition \ref{def41}). Under this constraint, Weak Commutativity would imply that $\succsim_{\bm{\mu}^E}$ is identical to $\succsim_{\bm{\mu}}^E$. In this case, $\succsim_{\bm{\mu}^E}$ is (weakly) more pessimistic than $\succsim_{\bm{\mu}}^E$, as is required in Pessimism to Update-then-Aggregate. On the other hand, when $|\Omega|\geq 4$, each of Axiom \ref{weak} and \ref{axiom4} neither implies nor is implied by the other.

The next result shows that Axiom \ref{weak} (Weak Commutativity) and Axiom \ref{axiom4} (Pessimism to Update-then-Aggregate) jointly characterize the multiple-weight aggregation rule, which also provides a characterization for MEU in the setting of belief aggregation. 

\begin{corollary}\label{coro42}
Suppose $|\Omega|\geq 4$. For a regular aggregation rule $\sigma$, the following statements are equivalent:
\begin{enumerate}
    \item Weak Commutativity and Pessimism to Update-then-Aggregate hold;
    \item Weak Commutativity holds and $\succsim_{\bm{\mu}}$ is ambiguity-averse for every $\bm{\mu}\in (\Delta\Omega)^n$;
    \item $\sigma$ is a multiple-weight aggregation rule;
    \item $\succsim_{\bm{\mu}}$ has an MEU representation for every $\bm{\mu}\in (\Delta\Omega)^n$.
\end{enumerate}
\end{corollary}

\subsection{Moderate Commutativity and Independence}

In this section, I introduce Axiom \ref{moderatecom}, a commutativity condition that jointly strengthens Axiom \ref{weak} (Weak Commutativity) and Axiom \ref{axiom4} (Pessimism to Update-then-Aggregate). I show that Axiom \ref{moderatecom} characterizes the classic Independence axiom in \cite{anscombe1963definition}.

\begin{axiom} \label{moderatecom} (\textit{Moderate Commutativity})
\label{MC}
For every $\bm{\mu}\in(\Delta\Omega)^n$  and $E\in\mathcal{E}(\bm{\mu})$ such that $\mu_1(E)= ...=\mu_n(E)$,
$$
f \succsim_{\bm{\mu}}^E g \iff f \succsim_{\bm{\mu}^E} g.
$$
\end{axiom}

Axiom \ref{moderatecom} (Moderate Commutativity) states that the DM's preference is uninfluenced by the order of updating and aggregation if all the experts agree about the probability of the realized event $E$. The axiom strengthens both Axiom \ref{weak} (Weak Commutativity) and Axiom \ref{axiom4} (Pessimism to Update-then-Aggregate). First, Axiom \ref{moderatecom} strengthens Axiom \ref{weak} because it drops the constraint on $\bm{\mu}$ that the experts must agree on the probabilities of \textit{every} state in $E^c$. Second, Axiom \ref{moderatecom} requires $\succsim_{\bm{\mu}}^E$ to be identical to $\succsim_{\bm{\mu}^E}$, which implies the latter being (weakly) more pessimistic than the former, as is required by Axiom \ref{axiom4}. Since Axiom \ref{moderatecom} rules out the possibility that $\succsim_{\bm{\mu}^E}$ is strictly more pessimistic than $\succsim_{\bm{\mu}}^E$, it captures the neutral attitude to the uncertainty about the reason for \textit{ex-post} disagreement. 

The next result shows that Moderate Commutativity is equivalent to the standard Independence axiom when $|\Omega|\geq 4$. A preference relation $\succsim$ satisfies \textit{Independence} if for every $f,g,h\in\mathcal{F}$ and $\alpha\in (0,1)$, $f\succsim g \iff \alpha f+(1-\alpha)h\succsim \alpha g+(1-\alpha)h$.

\begin{theorem} \label{thm42}
Suppose $|\Omega|\geq 4$. A regular aggregation rule satisfies Moderate Commutativity if and only if $\succsim_{\bm{\mu}}$ satisfies Independence for every $\bm{\mu}\in(\Delta\Omega)^n$.
\end{theorem}

Taking Theorem \ref{ambfree} and Theorem \ref{thm42} together leads to the following corollary.

\begin{corollary} \label{coro43}
Suppose $|\Omega|\geq 4$. For a regular aggregation rule $\sigma$, the following statements are equivalent:
\begin{enumerate}
    \item Moderate Commutativity holds;
    \item Weak Commutativity holds and $\succsim_{\bm{\mu}}$ satisfies P2 for every $\bm{\mu}$;
    \item $\succsim_{\bm{\mu}}$ satisfies Independence for every $\bm{\mu}$;
    \item The aggregation rule is a linear aggregation rule. 
\end{enumerate}
\end{corollary}

\subsection{An Impossibility Result}

I show that further strengthening of the commutativity condition leads to an impossibility result. An aggregation rule satisfies \textit{Commutativity} if for every $\bm{\mu}\in(\Delta\Omega)^n$ and $E\in\mathcal{E}(\bm{\mu})$,
$$
f \succsim_{\bm{\mu}}^E g \iff f \succsim_{\bm{\mu}^E} g.
$$
An aggregation rule is a {dictatorship} if there exists an expert $i\in N$ such that $\succsim_{\bm{\mu}}$ is represented by $EU_{\mu_i}$ for all $\bm{\mu}\in (\Delta\Omega)^n$.

\begin{theorem}
\label{thmimposs}
Suppose $|\Omega|\geq 4$. A regular aggregation rule satisfies \textit{Commutativity} if and only if it is a dictatorship.
\end{theorem}

Given Thought Experiment \ref{exp1} and \ref{exp2}, it should not be surprising that Commutativity is a strong assumption that can be inconsistent with many belief aggregation rules. The literature has also documented impossibility results implied by Commutativity. For example, \cite{genest1984conflict} showed that, for a family of pooling operators that maps suggestion profiles to the DM's beliefs, a version of Commutativity leads to dictatorship. Relatedly, \cite{chambers2010bayesian} shows that, for the problem of selecting beliefs from a set of possible priors, there does not exist a belief selection rule that satisfies an assumption analogous to Commutativity. Unlike the literature, Theorem \ref{thmimposs} takes a primitive where the DM forms arbitrary continuous and monotone preferences by aggregating expert suggestions. 

%% file: Sections/6.tex
\section{Discussion}
\label{sec6}

This section contains discussions of the fundamental assumptions of my model. It evaluates the two components of the regularity assumption, \textit{Pareto} and \textit{Monotonicity}.

The validity of \text{Pareto} has been questioned by the literature. \cite{gilboa2004utilitarian} presented a duel example where disagreement  over  beliefs and tastes ``cancel out'', generating spurious consensuses. In this paper, the DM interprets the belief $\mu_i$ suggested by each expert $i$ as the SEU function $\int_{\Omega} u\circ f d\mu_i$, where $u:X\rightarrow \mathbb{R}$ is the DM's utility index. As a result, there is no disagreement over tastes. \cite{mongin2016spurious} (pp. 524-525) presented an example where two experts agree that an event is more likely than another but the DM violates this ``consensus''. Mongin's example features unanimous priors and private signals, while the current paper focuses on heterogeneous priors and publicly observed event realizations.


For models of the aggregation of heterogeneous priors, \text{Pareto} can be violated if the DM is pessimistic about the \textit{overall} expertise of the experts, that is, if she thinks that she should be more pessimistic than all the experts $1,2,...,N$. Besides \text{Pareto}, the pessimism for the \text{overall} expertise of the experts can also lead to violations of \text{Monotonicity}. To illustrate this point more concretely, consider the state space $\Omega=\{G_1,G_2, B\}$, which represents the quality of an antique. The DM chooses between $f=(100, 100, -200)$ (to purchase the antique) and $x=(0,0,0)$. Consider the following two cases. In \textit{Case 1}, the experts unanimously agree that $Prob(G_1)=1$, the DM trusts the experts, and she prefers $f$ over $x$. In \textit{Case 2}, some experts believe $Prob(G_1)=1$ while others believe $Prob(G_2)=1$. While the experts still agree the value of $f$ is 100, the divergence in expert beliefs makes the DM concerned about the credibility of the consensus in the value of $f$. Such a concern drives the DM to choose the safe option $x$ instead of $f$. 

This example entails violations for both \text{Pareto} and \text{Monotonicity}. First, in \textit{Case 2}, all the experts suggest either $G_1$ or $G_2$, according to which the DM should choose $f$. However, the DM chooses the safe option $x$, which violates \text{Pareto}. Second, to see why \text{Monotonicity} is violated, notice that in both cases, all the experts unanimously agree that the value of $f$ is 100. If \text{Monotonicity} is true, this would imply that the DM attaches identical values to $f$ in both cases. This is violated in the example above because the DM prefers $f$ over $x$ in Case 1 but prefers $x$ over $f$ in Case 2. To my knowledge, the only paper in the literature on belief aggregation that focuses on related issues is \cite{bommier2021disagreement}. 

While \text{Pareto} and \text{Monotonicity} can be violated when the DM questions the overall expertise of the experts, these unanimity conditions seem more plausible in situations where the DM is sufficiently confident about the credibility of the panel of experts. When the experts $N$ are sufficiently knowledgeable about the state space $\Omega$, or when the experts adequately represent different opinions in the population, it becomes less attractive for the DM to form a preference that is more pessimistic than all knowledgeable experts or most members of society. In this case, the DM prefer option $f$ over option $x$ if all the experts believe that $f$ is more valuable than $x$. This alleviates the concerns for \text{Pareto} and \text{Monotoncity} raised above. 

%% file: Appendix/a.tex
\section{Proof of Theorem \ref{thmdual}}\label{seca}

The central mathematical object discussed in this paper is an aggregation rule $\sigma:(\Delta\Omega)^n\rightarrow\mathcal{R}$, which assigns a continuous and monotone weak order $\succsim_{\bm{\mu}}$ to every suggestion profile $\bm{\mu}$. 
This section presents the proof for Theorem 3.3, the main representation theorem in this paper. I start by proving the ``only if'' side. That is, I consider a fixed regular aggregation rule $\sigma$ that satisfies Weak Commutativity. I want to show the aggregation rule is a dual-self rule.

\subsection*{Step 1: Constructing the Aggregation Functional}

In the proofs of this paper, there are two important domains of functions: the set of experts $N$ and the set of states $\Omega$. To avoid confusion, I make the following distinctions in the notation system:
\begin{itemize}
    \item The set of experts is $N=\{1,2,,...,i,...,n\}$ with cardinality $|N|=n$.  $\mathbb{R}^N$ is the set of all $n$-dimentional real-valued vectors, with generic elements $a,b$. The set $\mathbb{R}^N$ is isomorphic to the set of real-valued functions with domain $N$, and I use these two terms interchangeably in the proof. For every $a\in\mathbb{R}^N$, define $\ub{a}=\min_{i\in N}a_i$ and $\bar{a}=\max_{i\in N}a_i$. The vector $a\in\mathbb{R}^N$ is constant if $\bar{a}=\ub{a}=c$ for some real number $c$. Define the constant vector as $c1_n$. When $c=1$, I define the constant vector as $1_n$. 
    \item The set of states is $\Omega=\{1,2,...,\omega,...\}$ with cardinality $|\Omega|\equiv l\geq 3$. For act $f$ and state $\omega$, I use $f_\omega$ to denote the payoff of $f$ given state $\omega$. For example, $f_1$ is the payoff in state $1$. $\mathbb{R}^\Omega$ is the set of all real-valued functions with domain $\Omega$, with generic elements $v,w$. I use $J$ to denote a generic functional that maps $\mathbb{R}^\Omega$ to $\mathbb{R}$. For every $v\in\mathbb{R}^\Omega$, define $\ub{v}=\min_{\omega\in \Omega}v_\omega$ and $\bar{v}=\max_{\omega\in \Omega}v_\omega$. The vector $v\in\mathbb{R}^\Omega$ is constant if $\bar{v}=\ub{v}=c$ for some real number $c$. Define the constant vector as $c1_\Omega$.  When $c=1$, I define the constant vector as $1_\Omega$.
\end{itemize}
In the first step, I establish the existence of an aggregation functional $I:\mathbb{R}^N\rightarrow \mathbb{R}$, such that for all $\bm{\mu}\in\Delta\Omega$, the corresponding weak order $\succsim_{\bm{\mu}}$ is represented by 
$$
I[EU_{\mu_1}(f),EU_{\mu_2}(f), \\...,EU_{\mu_n}(f)].
$$
I use $T_{n\times l}$ to define a (real-valued) matrix with $n$ rows and $l$ columns. Next, I state the definition for right stochastic matrices. 

\begin{definition}
The set of $n\times l$ right stochastic matrix by 
$$
\mathcal{T}=\{T_{n\times l}\geq 0\mid T \cdot 1_{l}=1_n\}.
$$
\end{definition}

Lemma \ref{A1} is a technical result that is useful for the construction of the aggregation functional $I:\mathbb{R}^N\rightarrow \mathbb{R}$.

\begin{lemma}
\label{A1}
For every function $\varsigma:\mathcal{T}\times\mathbb{R}^l\rightarrow \mathbb{R}$, if for every $T,T'\in\mathcal{T}$ and $v,w\in\mathbb{R}^l$,
    \begin{align}
    \label{condition1}
    \text{ if } T  v=T  w, \text{ then } \varsigma(T,v)=\varsigma(T,w),
    \end{align}
    and
    \begin{align}
    \label{condition2}
    \text{ if } T  v=T'  v, \text{ then } \varsigma(T,v)=\varsigma(T',v),
    \end{align}
then $\varsigma$ satisfies that, for every $T,T'\in \mathcal{T}$ and $v,w\in\mathbb{R}^l$,
\begin{align*}
T  v=T'  w\implies \varsigma(T,v)=\varsigma(T',w).\tag{\textasteriskcentered}
\label{main}
\end{align*}
\end{lemma}

The proof of Lemma \ref{A1} is less related to the central topic of this paper, and I delegate the proof to Appendix \ref{appendixc}. In Appendix \ref{appendixc}, I start with arbitrary stochastic matrices $T,T'\in \mathcal{T}$ and real vectors $v,w\in\mathbb{R}^l$, construct a sequence of stochastic matrices and vectors, apply (\ref{condition1}) and (\ref{condition2}) alternatively, and conclude that $\varsigma(T,v)=\varsigma(T',w)$, which is the desired result.

Here, I introduce some relevant notations for the next lemma. Recall that $u(\cdot): X\rightarrow \mathbb{R}$ is the linear utility function that represents $\succsim_0$, the DM's preference over the payoffs. For every act $f:\Omega\rightarrow \mathcal{F}$, I define $u_f\in\mathbb{R}^\Omega$ such that $u_f(\omega)=u(f_\omega)$ for every $\omega\in\Omega$. The vector $u_f$ is the ``utility act'' corresponding to $f$. Since I assume that $u(X)$ is unbounded, for every $v\in\mathbb{R}^\Omega$, there exists an act $f$ such that $u_f=v$. For every $\bm{\mu}\in (\Delta\Omega)^n$ and $f\in\mathcal{F}$, the inner product $\bm{\mu}\cdot u_f\in \mathbb{R}^N$ is defined by 
$$
\bm{\mu}\cdot u_f=(\mu_1\cdot u_f, \mu_2\cdot u_f,...,\mu_n\cdot u_f).
$$

\begin{lemma}\label{a2}
If the aggregation rule $\sigma$ is regular, then there exists a functional $I:\mathbb{R}^N\rightarrow \mathbb{R}$, such that for all $\bm{\mu}\in\Delta\Omega$, the function $U_{\bm{\mu}}:\mathcal{F}\rightarrow \mathbb{R}$ that represents $\succsim_{\bm{\mu}}$ is equal to 
$$
U_{\bm{\mu}}(f)=I(\bm{\mu}\cdot u_f).
$$
\end{lemma}

\begin{proof}
Recall from Section \ref{sec2} that, for every $\bm{\mu}\in (\Delta\Omega)^n$, the utility representation $U_{\bm{\mu}}(\cdot)$ of $\succsim_{\bm{\mu}}$ satisfies $U_{\bm{\mu}}(x)=u(x)$ for every $x\in X$. It is standard to verify that there exists a monotonic and continuous functional $\Phi_{\bm{\mu}}(\cdot):\mathbb{R}^\Omega\rightarrow \mathbb{R}$ such that $U_{\bm{\mu}}(f)=\Phi_{\bm{\mu}}(u_f)$ for every $f\in\mathcal{F}$. Define $\hat{I}(\cdot): (\Delta\Omega)^n\times\mathbb{R}^\Omega\rightarrow \mathbb{R}$ by
$
\hat{I}(\bm{\mu},v)=\Phi_{\bm{\mu}}(v)
$  for every $\bm{\mu}\in (\Delta\Omega)^n$ and $v\in \mathbb{R}^\Omega$.
Suppose $\sigma$ is regular. For every $\bm{\mu},\bm{\mu}'\in\Delta(\Omega)$ and $a,b\in\mathbb{R}^N$, \textit{Pareto} and \textit{Monotonicity} imply the following two statements respectively:
\begin{enumerate}
    \item If
        $
        \bm{\mu}\cdot v\geq  \bm{\mu}\cdot w$, then $\hat{I}(\bm{\mu}, v)\geq \hat{I}(\bm{\mu}, w),
        $ and
    \item if
        $\bm{\mu}\cdot v\geq  \bm{\mu}'\cdot w$, then $\hat{I}(\bm{\mu}, v)\geq \hat{I}(\bm{\mu}', w).
        $
\end{enumerate}
The above two statements still hold if we alternate all the ``$\geq$'' above to ``$=$''. A direct application of Lemma \ref{A1} leads to the following result: for every $\bm{\mu},\bm{\mu}'\in(\Delta\Omega)^n$ and $v,w\in\mathbb{R}^\Omega$,
\begin{align} \label{a2eq1}
        \bm{\mu}\cdot v=  \bm{\mu}'\cdot w\implies \hat{I}(\bm{\mu}, v)= \hat{I}(\bm{\mu}', w).
\end{align}
Now, for every $a\in \mathbb{R}^N$, define $I(a)=\hat{I}(\bm{\mu},v)$ for $\bm{\mu}\in (\Delta\Omega)^n$ and $v\in\mathbb{R}^\Omega$ such that $\bm{\mu}\cdot v=a$. First, such $\bm{\mu}$ and $v$ must exist.\footnote{Fix an arbitrary $a\in \mathbb{R}^N$. If $a$ is a constant vector, that is, if $a=c1_\Omega$ for some $c$, then any belief profile $\bm{\mu}$ and $v=c1_n$ will satisfy $a=\bm{\mu}\cdot v$. If $a$ is not a constant vector, then construct $\bm{\mu}\in(\Delta\Omega)^n$ and $v\in\mathbb{R}^\Omega$ such that (1) $\mu^*_i(\omega=1)=\frac{a_i}{\bar{a}-\ub{a}}$ for every $i\in N$, (2) $v^*_1=\bar{a}-\ub{a}$, and (3) $v^*_\omega=0$ for every $\omega\neq 1$, then it holds that $a=\bm{\mu}^*\cdot v^*$.} Second, (\ref{a2eq1}) guarantees the well-definedness of $I:\mathbb{R}^N\rightarrow \mathbb{R}$. Finally, the definitions of $I$ and $\hat{I}$ imply $U_{\bm{\mu}}(f)=I(\bm{\mu}\cdot u_f)$, the desired result. 
\end{proof}

I call $I:\mathbb{R}^N\rightarrow \mathbb{R}$ the \textit{aggregation functional}.

\subsection*{Step 2: Monotonicity and Constant Linearity}

The next lemma establishes that under Axiom \ref{weak}, the aggregation functional $I$ must be monotonic and constant linear.

\begin{lemma} \label{a3}
    For every $a,b\in \mathbb{R}^N$, $\beta\in (0,1)$ and $c\in\mathbb{R}$, 
    \begin{enumerate}
        \item (Monotonic Functional) if $a\leq b$, then $I(a)\leq I(b)$; 
        \item (Constant Linear Functional) $I(\beta a+(1-\beta)c1_n)=\beta I(a)+(1-\beta) c$.
    \end{enumerate}
\end{lemma}

\begin{proof}
First, suppose $a\leq b$ and I want to show $I(a)\leq I(b)$.  The case where $a=b$ is trivial. For the case where $a\neq b$, it must hold that $\bar{b}-\ub{a}\neq 0$ or $\bar{a}-\ub{b}\neq 0$, otherwise, $a$ and $b$ are two constant vectors that are equal to each other. Without loss of generality, I assume $\bar{b}-\ub{a}\neq 0$. Construct act $f^*$ such that that $u_{f^*}=(\ub{a},\bar{b},0,0,...,0)$. Construct $\bm{\mu}^*, \bm{\mu}^{**}$ such that for every $i=1,...,n$, 
$
\mu_i^*=(\frac{\bar{b}-a_i}{\bar{b}-\ub{a}},1-\frac{\bar{b}-a_i}{\bar{b}-\ub{a}},0,...,0)
$
and
$
\mu_i^{**}=(\frac{\bar{b}-b_i}{\bar{b}-\ub{a}},1-\frac{\bar{b}-b_i}{\bar{b}-\ub{a}},0,...,0)
$. 
Since I assume $\bar{b}-\ub{a}\neq 0$, the utility act $u_{f^*}$ and the suggestion profiles $\bm{\mu}^*$ and $\bm{\mu}^{**}$ are well-defined. By calculating the inner product, we may conclude that $\bm{\mu}^* \cdot u_{f^*}=a$ and $\bm{\mu}^{**} \cdot u_{f^*}=b$. By \textit{Monotonicity} in Assumption \ref{assumption1}, $\hat{I}(\bm{\mu}^*,f)\leq \hat{I}(\bm{\mu}^{**},f)$, which implies $
I(a)\leq I(b) 
$.

Second, in order to prove the second part of Lemma \ref{a3}, I want to first show $I(c1_n)=c$ for all $c\in\mathbb{R}$. Let $x\in X$ be a payoff such that $u(x)=c$. Such an $x$ exists due to the unboundedness of $u$. For every $\bm{\mu}\in (\Delta\Omega)^n$, by the definition of $U_{\bm{\mu}}:\mathcal{F}\rightarrow \mathbb{R}$, it holds that $U_{\bm{\mu}}(x)=u(x)=c$. As a result, for every $\bm{\mu}\in(\Delta\Omega)^n$, $I(\bm{\mu}\cdot u_x)=c$. Notice that $u_x=c1_\Omega$ and $\bm{\mu}\cdot u_x=c1_n$. As a result, it holds that $I(c1_\Omega)=c$.

Now, I want to show $I(\beta a+(1-\beta)c1_n)=\beta I(a)+(1-\beta) c$. Without loss of generality, assume $\ub{a}<\bar{a}$.\footnote{If $\ub{a}=\bar{a}$, then $a=\ub{a}1_n$ and $I(\beta \ub{a}1_n+(1-\beta)c1_n)=\beta\ub{a}+(1-\beta) c$ must hold, which implies the desired result.} Construct $\hat{f}$ such that $u_{\hat{f}}=(\bar{a},\ub{a},c,0,...,0)$. Construct $\hat{\bm{\mu}}$ such that
\begin{align} \label{construct}
\hat{\mu}_{i}=\left(\beta\left(\frac{a_i-\ub{a}}{\bar{a}-\ub{a}}\right), \beta\left(1-\frac{a_i-\ub{a}}{\bar{a}-\ub{a}}\right), 1-\beta,  0,...,0\right)
\end{align} 
for every expert $i\in N$. It can be verified by calculating the inner product that $\hat{\bm{\mu}}\cdot u_{\hat{f}}=\beta a+(1-\beta)c1_n$. 

Next, I apply Axiom \ref{weak} (Weak Commutativity) to conclude that $I(\cdot)$ is constant linear. To apply Axiom \ref{weak}, consider an event $G$ such that $\{1,2\}\subset G$ and $\{3\}\not\subset G$. Consider a DM who updates on $G$. Define the complement of $G$ by $H$. By construction, $\{3\}\subset H$. By (\ref{construct}), $\hat{\bm{\mu}}$ is a suggestion profile for which disagreement is restricted within $G$, so Axiom 1 (Weak Commutativity) can be applied. Consider the conditional preferences $\succsim_{\hat{\bm{\mu}}}^G$ and $\succsim_{\hat{\bm{\mu}}}^{H}$. By Weak Commutativity, $\succsim_{\hat{\bm{\mu}}}^G$ is equal to $\succsim_{\hat{\bm{\mu}}^G}$. Now I argue that $\succsim_{\hat{\bm{\mu}}}^{H}$ is equal to $\succsim_{\hat{\bm{\mu}}^H}$. First, notice from (\ref{construct}) that $\hat{\mu_i}^H=(0,0,1,0,...,0)$ for all $i\in N$. By Pareto in Assumption \ref{assumption1}, $\succsim_{\hat{\bm{\mu}}^H}$ is represented by $U(f)=u(f_3)$, where $f_3$ is the payoff of $f$ given $\omega=3$. Now, it suffices to show that $\succsim_{\hat{\bm{\mu}}}^H$ is represented by $U(f)=u(f_3)$. 
\begin{enumerate}
    \item First, I want to show that $u(f_3)\geq u(g_3)$ implies $f\succsim_{\hat{\bm{\mu}}}^Hg$. Notice that, since $\hat{\mu}_i^H=(0,0,1,0,...)$ for every $i$, $\omega=3$ is the only non-null state in $H$. Therefore, for every $i\in N$,  $u(f_3)\geq u(g_3)$ if and only if $EU_{\hat{\mu}_i}[fHh]\geq EU_{\hat{\mu}_i}[gHh]$. As a result, by Pareto, if $u(f_3)\geq u(g_3)$, then $fHh\succsim_{\hat{\bm{\mu}}}gHh$. This holds for all $h$, which implies $f\succsim_{\hat{\bm{\mu}}}^Hg$.
    \item Second, I want to show that $u(f_3)> u(g_3)$ implies $g\not{\succsim}_{\hat{\bm{\mu}}}^H f$. Notice that, by the construction in (\ref{construct}), for every $x\in X$ and $i\in N$, it holds that $\hat{\mu}_i\cdot u_{fHx}=\beta u(x)+(1-\beta)u(f_3)$. As a result, $\hat{\bm{{\mu}}}\cdot u_{fHx}=[\beta u(x)+(1-\beta)u(f_3)]1_{n}$. In the proof for Lemma \ref{a3}, I show that $I(c1_n)=c$ for any $c$. As a result, $I(\hat{\bm{\mu}}\cdot u_{fHx})=I\{[\beta u(x)+(1-\beta)u(f_3)]1_{n}\}=\beta u(x)+(1-\beta)u(f_3)$. Similarly, $I(\hat{\bm{\mu}}\cdot u_{gHx})=\beta u(x)+(1-\beta)u(g_3)$. Since $u(f_3)>u(g_3)$, it holds that $I(\hat{\bm{\mu}}\cdot u_{fHx})>I(\hat{\bm{\mu}}\cdot u_{gHx})$, which implies $fHx\succ_{\hat{\bm{\mu}}} gHx$, i.e. $gHx\not{\succsim}_{\hat{\bm{\mu}}}fHx$. By the definition of conditional preference, it holds that $g\not{\succsim}_{\hat{\bm{\mu}}}^H f$. 
\end{enumerate}

As a result, $\succsim_{\hat{\bm{\mu}}}^H$ is also represented by $U(f)=u(f_3)$, which implies $\succsim_{\hat{\bm{\mu}}}^H$ is identical to $\succsim_{\hat{\bm{\mu}}^H}$.

The analysis above establishes that $\succsim_{\hat{\bm{\mu}}}^G$ is equal to $\succsim_{\hat{\bm{\mu}}^G}$ and $\succsim_{\hat{\bm{\mu}}}^H$ is identical to $\succsim_{\hat{\bm{\mu}}^H}$. Next, I want to show that the aggregation functional $I$  is constant linear. Let $y,z\in X$ be two payoffs such that $\hat{f}\sim_{\hat{\bm{\mu}}}^G y$ and $\hat{f}\sim_{\hat{\bm{\mu}}}^H z$.\footnote{There must exist $y$ such that $\hat{f}\sim_{\hat{\bm{\mu}}}^G y$. This is because Weak Commutativity implies $\succsim_{\hat{\bm{\mu}}}^G$ is equal to $\succsim_{\hat{\bm{\mu}}^G}$, where the latter admits a continuous utility representation $U_{{\hat{\bm{\mu}}^G}}$. Therefore, there exists a $y$ such that $u(y)=U_{{\hat{\bm{\mu}}^G}}(\hat{f})$. The same argument works if $G$ is altered by $H$ and $y$ is altered by $z$.} Since $H=G^c$, the definition of conditional preference implies $\hat{f}Gz\sim_{\hat{\bm{\mu}}} yGz$ and $\hat{f}\sim_{\hat{\bm{\mu}}}\hat{f}Gz$. Then, transitivity implies $\hat{f}\sim_{\hat{\bm{\mu}}} yGz$, that is, $I(\hat{\bm{\mu}}\cdot u_{\hat{f}})=I(\hat{\bm{\mu}}\cdot u_{yGz})$. Notice that $\hat{\mu}_i(G)=\beta$ for all $i\in N$ by how $\hat{\bm{\mu}}$ is constructed. As a result, $\hat{\bm{\mu}}\cdot u_{yGz}=[\beta u(y)+(1-\beta)u(z)]\cdot 1_n$. In the proof for Lemma \ref{a3}, I show that $I(c1_n)=c$ for any $c$, therefore, $I(\hat{\bm{\mu}}\cdot u_{yGz})=\beta u(y)+(1-\beta)u(z)$. Since $\hat{f}\sim_{\hat{\bm{\mu}}} yGz$, it holds that $I(\hat{\bm{\mu}}\cdot u_{\hat{f}})=\beta u(y)+(1-\beta)u(z)$. Recall that $\hat{\bm{\mu}}\cdot u_{\hat{f}}=\beta a+(1-\beta)c1_n$ by how $\hat{\bm{\mu}}$ and $\hat{f}$ are constructed. Therefore, $I[\beta a+(1-\beta)c1_n]=\beta u(y)+(1-\beta)u(z)$.


Finally, since $\hat{f}\sim_{\hat{\bm{\mu}}}^G y$ and $\hat{f}\sim_{\hat{\bm{\mu}}}^H z$, it holds that $\hat{f}\sim_{\hat{\bm{\mu}}^G} y$ and $\hat{f}\sim_{\hat{\bm{\mu}}^H} z$. That is, $I(\hat{\bm{\mu}}^G\cdot u_{\hat{f}})=u(y)$ and $I(\hat{\bm{\mu}}^H\cdot u_{\hat{f}})=u(z)$. Applying Bayes' rule to the $\hat{\mu}_i$ defined in (\ref{construct}), we can conclude by algebraic calculation that $\hat{\bm{\mu}}^G\cdot u_{\hat{f}}=a$ and $\hat{\bm{\mu}}^{H}\cdot u_{\hat{f}}=c1_n$. Therefore, $I(a)=u(y)$ and $I(c1_n)=u(z)$. It follows that $I[\beta a+(1-\beta)c1_n]=\beta I(a)+(1-\beta)I(c1_n)=\beta I(a)+(1-\beta)c$, which is the desired result.
\end{proof} 

The lemmas above establish that if a regular aggregation rule satisfies Weak Commutativity, the corresponding aggregation functional $I:\mathbb{R}^N\rightarrow\mathbb{R}$ must be monotonic and constant linear. Applying the main result in \cite{chandrasekher2022dual} implies the existence of a weight-set collection $\bm{\Lambda}\in \mathcal{K}$ such that, for all $a\in \mathbb{R}^N$, 
$$
I(a)=\max_{\Lambda\in\bm{\Lambda}} \min_{\lambda\in\Lambda} \lambda\cdot a. 
$$
Substitute this to $U_{\bm{\mu}}(f)=I(\bm{\mu}\cdot u_f)$, then 
\begin{align}
U_{\bm{\mu}}(f)&=\max_{\Lambda\in\bm{\Lambda}} \min_{\lambda\in\Lambda} \lambda\cdot \bm{\mu}\cdot u_f\\    
&= \max_{\Lambda\in\bm{\Lambda}} \min_{\lambda\in\Lambda}\int_{\Omega} u\circ f d p_{\lambda}(\bm{\mu}),\label{dself}
\end{align}
where $p_{\lambda}(\bm{\mu})=\sum_{i}\lambda_i \mu_i$. This completes the proof for the ``only if'' side of Theorem 3.3. 

\subsection*{Step 3: the ``if'' Side}

For the ``if'' side, consider a dual-self aggregation rule that takes the analytic form (\ref{dself}). For a suggestion profile $\bm{\mu}$ such that expert disagreement is restricted within $E$, we want to show that $\succsim_{\bm{\mu}}^E$ and $\succsim_{\bm{\mu}^E}$ are identical. 

Recall that disagreement is restricted within $E$ implies $\mu_1(\omega)=...=\mu_n(\omega)$ for every $\omega\in E^c$. In this case, all the experts agree about the belief conditional on $E^c$, that is, $\mu_1^{E^c}=...\mu_n^{E^c}$; denote this conditional distribution by $\mu_0\in\Delta(E^c)$. Also, all the experts agree about the probability of event $E$, a real number denoted as $\alpha\in (0,1)$. It follow that $\mu_i=\alpha\mu_i^E+(1-\alpha)\mu_0$ for all $i$. Substitute this to expression (\ref{dself}), it follows that for every $f\in\mathcal{F}$,
\begin{align*}
U_{\bm{\mu}}(f)
&=
\alpha \max_{\Lambda\in\bm{\Lambda}} \min_{\lambda\in\Lambda}\int_{\Omega} u\circ f d p_{\lambda}(\bm{\mu}^E)+(1-\alpha)\int_{\Omega} u\circ fd \mu_0\\
&=
\alpha U_{\bm{\mu}^E}(f)+(1-\alpha)\int_{\Omega} u\circ fd \mu_0.
\end{align*}
Since $\mu_0=\mu_i^{E^c}$, it holds that $\mu_0(\omega)=0 \ \forall \omega\in E$. As a result, for every $g,h\in\mathcal{F}$, 
$$
U_{\bm{\mu}}(fEh)=\alpha U_{\bm{\mu}^E}(f)+(1-\alpha) \int_{\Omega} u\circ hd \mu_0, \text{ and }
$$
$$
U_{\bm{\mu}}(gEh)=\alpha U_{\bm{\mu}^E}(g)+(1-\alpha) \int_{\Omega} u\circ h d \mu_0.
$$
Since $\alpha>0$, it holds that $fEh\succsim_{\bm{\mu}}gEh \iff U_{\bm{\mu}}(fEh)\geq U_{\bm{\mu}}(gEh) \iff U_{\bm{\mu}^E}(f)\geq U_{\bm{\mu}^E}(g)$, where the last inequality is equivalent to $f\succsim_{\bm{\mu}^E}g$. By definition of conditional preference, it holds that $f\succsim_{\bm{\mu}}^E g \iff f\succsim_{\bm{\mu}^E}g$, the desired result.

%% file: Appendix/b.tex
\section{Proofs for Sections 3, 4, and 5}
\label{secb}
\subsection{Proof of Theorem \ref{ambfree}}\label{secb1}

The ``if'' side of the theorem is standard and I omit it. For the ``only if'' side, I consider a regular aggregation rule that satisfies Weak Commutativity and I assume $\succsim_{\bm{\mu}}$ satisfies STP for every $\bm{\mu}\in(\Delta\Omega)^n$. I restrict attention to the situation where $|\Omega|=\{1,2,3\}$ and prove that the aggregation rule $\sigma$ is linear for this case. For state spaces with greater cardinalities, I may restrict attention to belief profiles with support over the first three states, and the proof below can be replicated to establish that the aggregation functional $I$ must be linear, which implies that $\sigma$ is a linear aggregation rule.

\begin{lemma}\label{b1}
    Suppose $\Omega=\{1,2,3\}$, then for every $\bm{\mu}\in(\Delta\Omega)^n$ such that $\mu_i(\omega)>0$ for every $\omega\in\Omega$ and $i\in N$, the utility representation $U_{\bm{\mu}}(\cdot)$ must admit a SEU representation, that is, there exists $\mu_0\in \Delta\Omega$ such that, for every $f\in\mathcal{F}$,
    $$
    U_{\bm{\mu}}(f)=EU_{\mu_0}(f). 
    $$
\end{lemma}

\begin{proof}
    First, since the aggregation rule is regular and satisfies Weak Commutativity, for every $\bm{\mu}\in (\Delta\Omega)^n$, there exists a unique weight-set collection $\bm{\Lambda}$ such that the utility function equals 
    \begin{align} \label{b1dual}
        U_{\bm{\mu}}(f)=\max_{\Lambda\in\bm{\Lambda}} \min_{\lambda\in\Lambda}\int_{\Omega} u\circ f d p_{\lambda}(\bm{\mu}).
    \end{align}
    
    For every $\Lambda\in\bm{\Lambda}$, define $P_{{\Lambda}}(\bm{\mu})=\{p_{\lambda}(\bm{\mu})\mid \lambda\in \Lambda\}\subset\Delta\Omega$ and $\mathbb{P}_{\bm{\Lambda}}(\bm{\mu})=\{P_{\Lambda}(\bm{\mu})\mid \Lambda\in\bm{\Lambda}\}$. When $\bm{\Lambda}$ and $\bm{\mu}$ are specified in the context, I simply denote $\mathbb{P}_{\bm{\Lambda}}(\bm{\mu})$ as $\mathbb{P}$. Based on these definitions, (\ref{b1dual}) can be rewritten as 
    \begin{align}\label{dualutil}
    U_{\bm{\mu}}(f)=\max_{P\in \mathbb{P}}\min_{p\in P} \int_{\Omega} u\circ f dp, 
    \end{align}
    which is a DSEU representation as in \cite{chandrasekher2022dual}. Define the functional $J_{\bm{\mu}}:\mathbb{R}^\Omega\rightarrow \mathbb{R}$ by $J_{\bm{\mu}}(v)=\max_{P\in \mathbb{P}}\min_{p\in P}  p\cdot v$. Since $\mu_i(\omega)>0$ for every $\omega\in\Omega$ and $i\in N$, by \textit{Pareto}, every set of beliefs in $\mathbb{P}$ must lie in the interior of $\Delta\Omega$. As a result, $J_{\bm{\mu}}$ must be strongly monotonic, that is, for every $w\gneqq w'$, $J_{\bm{\mu}}(w)>J_{\bm{\mu}}(w')$.
    
    Denote $F=\{1,2\}$ and $F'=\{3\}$, and $\{F,F'\}$ forms a partition for the state space. Since the preference $\succsim_{\bm{\mu}}$ is assumed to satisfy $STP$, the conditional preferences $\succsim_{\bm{\mu}}^F$ and $\succsim_{\bm{\mu}}^{F'}$ are both well-defined. According to Theorem S.2.1 in the supplementary material of \cite{chandrasekher2022dual}, there exists monotonic and constant linear functionals $J_0:\mathbb{R}^\Pi\rightarrow\mathbb{R}$ and $J_F:\mathbb{R}^{F}\rightarrow\mathbb{R}$, such that for every $f=(f_1,f_2,f_3)\in\mathcal{F}$, 
\begin{enumerate}
    \item $\hat{U}(f)=J_F[u(f_1),u(f_2)]$ represents  $\succsim_{\bm{\mu}}^F$;
    \item $J_{\bm{\mu}}(v_1,v_2,v_3)=J_0[J_{F}(v_1,v_2),v_3]$.\footnote{Specifically, in the supplementary material of \cite{chandrasekher2022dual}, the existence of $J_0$ and that $J_0$ is monotonic and linear are established in the paragraph under expression (29) on page 6.}
\end{enumerate}
Moreover, since $J_{\bm{\mu}}$ is strongly monotonic, $J_0$ and $J_F$ are both strongly monotonic. Now let $c\in (0,1)$ be such that 
\begin{align} \label{B1eq1}
    J_{\bm{\mu}}(1,1,0)=J_{\bm{\mu}}(1,c,c),
\end{align} and construct $f^*,x^*\in\mathcal{F}$ such that $u_{f^*}=(1,1,0)$ and $u(x^*)=c$.\footnote{The existence of $c\in [0,1]$ is guaranteed by the continuity of $U_{\bm{\mu}}$ (and hence $J_{\bm{\mu}}$) and the intermediate value theorem. The strong monotonicity of $J_{\bm{\mu}}$ then implies $c\in (0,1)$.} Let $H=\{1\}$. Then by how $f^*$ and $x^*$ are constructed, (\ref{B1eq1}) implies that $f^*\sim_{\bm{\mu}} f^*Hx^*$. Therefore, by $STP$, $x^* H f^*\sim_{\bm{\mu}}x^*$, i.e. 
\begin{align} \label{B1eq2}
J_{\bm{\mu}}(c,1,0)=c.
\end{align}

To continue the proof, notice the fact that for any strongly monotonic and constant linear function $J^*:\mathbb{R}^2\rightarrow \mathbb{R}$, either $J^*(v_1,v_2)=\max_{t\in T}tv_1+(1-t)v_2$  for a closed set $T\subset (0,1)$, or
$J^*(v_1,v_2)=\min_{t\in T'}tv_1+(1-t)v_2$ for a closed set $T'\subset (0,1)$.  To avoid digression, I state the proof of this fact in footnote \ref{footnote17} below.\footnote{\label{footnote17}This fact is known in the literature, for example, see footnote 19 of \cite{chandrasekher2022dual}. Here, since $J^*(\cdot)$ is strongly monotone and constant linear, there exists $t,t'\in (0,1)$ such that 
$
J^*(v_1,v_2)=\begin{cases}
    t v_1+(1-t)v_2 \text{ if } v_1\geq v_2;\\
    t'v_1+ (1-t')v_2. \text{ if } v_1< v_2.
\end{cases}
$
Let $T^*=\{t,t'\}$.  If $t\geq t'$, then  $J^*(v_1,v_2)=\max_{t''\in T^*}t''v_1+(1-t'')v_2$ for every $v\in\mathbb{R}^2$; otherwise,  $J^*(v_1,v_2)=\min_{t''\in T^*}t''v_1+(1-t'')v_2$ for every $v\in\mathbb{R}^2$.}  Then, without loss of generality, assume $J_F(v_1,v_2)=\max_{t\in [\ub{t},\bar{t}]}tv_1+(1-t)v_2$, where $0<\ub{t}\leq \bar{t}<1$. Substitute this into (\ref{B1eq1}) and (\ref{B1eq2}), then we have (a) $J_0(1,0)=J_0(\bar{t}+(1-\bar{t})c,c)$ and (b) $J_0(\ub{t}c+(1-\ub{t}),0)=c$. Subtract (b) from (a) on both sides of the equation, we have $J_0(1,0)-J_0(\ub{t}c+(1-\ub{t}),0)=J_0(\bar{t}+(1-\bar{t})c,c)-c$. By constant linearity of $J_0$, the left-hand side of the equation is equal to $J_0[(1,0)-(\ub{t}c+(1-\ub{t}),0)]=J_0(\ub{t}(1-c),0)$, while the right-hand side of the equation is equal to $J_0(\bar{t}(1-c),0)$. As a result, $J_0(\ub{t}(1-c),0)=J_0(\bar{t}(1-c),0)$, which, by strong monotonicity of $J_0$, implies $\ub{t}=\bar{t}\equiv t^*$. Therefore, $J_F(v_1,v_2)=t^* v_1+(1-t^*)v_2$. 

Finally, notice that either of the following two cases is true: (i) there exists a closed set $S\subset (0,1)$, such that $J_0(w_1,w_2)=\min_{s\in S} sw_1+(1-s)w_2$, (ii) there exists a closed set $S'\subset (0,1)$, such that $J_0(w_1,w_2)=\max_{s\in S'} sw_1+(1-s)w_2$. Suppose without loss of generality that case (i) is true, which implies $\succsim_{\bm{\mu}}$ admits a maxmin expected utility representation. In this case, since $J_{\bm{\mu}}(v_1,v_2,v_3)=J_0[J_{F}(v_1,v_2),v_3]$, it holds that $J_{\bm{\mu}}(v)=\min_{p\in P^*} p\cdot v$, where $P^*=\left\{p=(st^*, s(1-t^*),1-s)\mid s\in S\right\}$. Now, relabel the states so that state $1$ in the analysis above is labeled as state $2$, state $2$ is labeled as state $3$, and state $3$ is labeled as state 1. Repeat the analysis above for the relabeled state space, and we can conclude that $J_{\bm{\mu}}(v)=\min_{p\in P^{**}}p\cdot v$, where $P^{**}=\{(1-s,st^{**},s(1-t^{**}))\mid s\in S'\}$ for some closed set $S'\subset (0,1)$ and $t^{**}\in (0,1)$. The uniqueness of maxmin expected utility representation implies $P^*=P^{**}$, which implies $p(\omega=1)/p(\omega=2)={t^*}/(1-t^*)$ and $p(\omega=2)/p(\omega=3)=t^{**}/(1-t^{**})$ for all $p\in P^*$. As a result, $P^*$ must be a singleton, and denote $P^*$ by $\{\mu_0\}$. It follows that $J_{\bm{\mu}}(v)=\mu_0 v$ and     $
    U_{\bm{\mu}}(f)=EU_{\mu_0}(f),
    $
which is the desired result.
\end{proof}

With the preparation of Lemma \ref{b1}, we are ready to prove Theorem 3.1. Suppose towards a contradiction that $\sigma$ is not a linear aggregation rule. Then the corresponding aggregation functional $I:\mathbb{R}^N\rightarrow \mathbb{R}$ is also non-linear. That is, there exists $a,b\in\mathbb{R}^N$ and $\beta\in (0,1)$ such that $I(\beta a+(1-\beta) b)\neq \beta I(a)+(1-\beta)I(b)$. Since $I$ is constant linear, without loss of generality we can assume $0< a,b< \frac{1}{2}1_{n}$.\footnote{Specifically, take a constant $c>\max_{i}\{|a_i|,|b_i|\}\in \mathbb{R}$. Then take $a'=(a+c1_n)/4c\in\mathbb{R}^N$ and $b'=(b+c1_n)/4c\in\mathbb{R}^N$; it is easy to verify that $0< a,b< \frac{1}{2}1_{n}$. By constant linearity we can conclude $I(\beta a'+(1-\beta)b')\neq \beta I(a')+(1-\beta)I(b')$.} Then we construct the belief profile $\bm{\mu}^*$ such that $\mu_J^*(\omega=1)=a_i>0$, $\mu_J^*(\omega=2)=b_i>0$ and $\mu_J^*(\omega=3)=1-a_i-b_i>0$ for every $i\in N$. Consider acts $f^*,g^*\in \mathcal{F}$ such that $u_{f^*}=(1,0,0)$ and $u_{g^*}=(0,1,0)$. By how $\bm{\mu}^*\in(\Delta\Omega)^n$ and $f^*,g^*\in\mathcal{F}$ are constructed, it holds that $\bm{\mu}^*\cdot u_{f^*}=a$ and $\bm{\mu}^*\cdot u_{g^*}=b$. Then 
\begin{align*}
U_{\bm{\mu}^*}(\beta f^*+(1-\beta)g^*)
&=
I(\beta \bm{\mu}^*\cdot u_{f^*}+(1-\beta)\bm{\mu}^*\cdot u_{g}^*)\\
&=I(\beta a+(1-\beta)b)\\
&\neq
\beta I(a)+(1-\beta)I(b)\\
&=\beta U_{\bm{\mu}^*}(f^*)+(1-\beta)U_{\bm{\mu}^*}(g^*).
\end{align*}
That is, the function $U_{\bm{\mu}^*}:\mathcal{F}\rightarrow\mathbb{R}$ is \textit{not} mixture linear. Therefore, it cannot be the case that $U_{\bm{\mu}^*}(f)=EU_{\mu_0}(f)$ for some $\mu_0\in\Delta\Omega$, which is contradictory to Lemma \ref{b1}.

\subsection{Proof for Theorem \ref{ambaverse}}\label{secb2}

The ``if'' side of Theorem \ref{ambaverse} is standard and I omit it. For the ``only if'' side, consider a regular aggregation rule such that $\succsim_{\bm{\mu}}$ is ambiguity averse for every $\bm{\mu}\in (\Delta\Omega)^n$. Suppose the aggregation rule satisfies Weak Commutativity, then according to Lemma \ref{a2} in Appendix \ref{seca}, there exists a monotonic and constant linear aggregation functional $I:\mathbb{R}^N\rightarrow \mathbb{R}$ such that for all $\bm{\mu}\in\Delta\Omega$, 
$
U_{\bm{\mu}}(f)=I(\bm{\mu}\cdot u_f).
$

\begin{lemma} \label{b2}
    The aggregation functional $I:\mathbb{R}^N\rightarrow \mathbb{R}$ corresponding to the aggregation rule is quasi-concave, that is, for every $a,b\in\mathbb{R}^N$ and $\alpha\in (0,1)$, if $I(a)=I(b)$, then $I(\alpha a+(1-\alpha)b)\geq I(a)$.
\end{lemma}
\begin{proof}
    Suppose $I$ is not quasi-concave, that is, there exists $a^*,b^*\in\mathbb{R}^N$ and $\alpha\in (0,1)$ such that $I(a^*)=I(b^*)$ and $I(\alpha a^*+(1-\alpha)b^*)< I(a^*)$. Similar to the proof in Section \ref{secb1}, since $I$ is constant linear, without loss of generality we can assume $0<a^*, b^*<\frac{1}{2}1_n$. Then I construct the belief profile $\bm{\mu}^*$ such that $\mu_i^*(\omega=1)=a^*_i$, $\mu_i^*(\omega=2)=b^*_i$ and $\mu_i^*(\omega=3)=1-a^*_i-b^*_i$ for every $i\in N$. Consider acts $f^*,g^*\in \mathcal{F}$ such that $u_{f^*}=(1,0,0,...)$ and $u_{g^*}=(0,1,0,...)$. It holds that $\bm{\mu}^*\cdot u_{f^*}=a^*$ and $\bm{\mu}^*\cdot u_{g^*}=b^*$. Then 
    \begin{align*}
U_{\bm{\mu}^*}(\alpha f^*+(1-\alpha)g^*)
&=
I(\alpha \bm{\mu}^*\cdot u_{f^*}+(1-\alpha)\bm{\mu}^*\cdot u_{g^*})\\
&=I(\alpha a^*+(1-\alpha)b^*)\\
&<
I(a^*)=U_{\bm{\mu}}(f^*).
\end{align*}
Also, $U_{\bm{\mu}}(f^*)=I(a^*)=I(b^*)=U_{\bm{\mu}}(g^*)$.  As a result, $f^*\sim_{\bm{\mu}}g^*$ and $\alpha f^*+(1-\alpha)g^*\prec_{\bm{\mu}}f^*$. This is contradictory to the fact that $\succsim_{\bm{\mu}}$ satisfies ambiguity aversion. 
\end{proof}

As a result of Lemma \ref{b2}, $I:\mathbb{R}^N\rightarrow\mathbb{R}$ is a monotonic, constant linear and quasi-concave. The main result of \cite{gilboa1989maxmin} then implies the existence of a closed and convex set of weights $\Lambda$ such that for every $a\in \mathbb{R}^N$, $I(a)=\min_{\lambda\in \Lambda}\lambda\cdot a$. It follows that the aggregation rule is a multiple-weight rule. 

\subsection{Proof of Theorem \ref{thm41}}

The ``only if'' side is a corollary of Theorem \ref{thmdual}. Consider a regular aggregation rule that satisfies Weak Commutativity, then Theorem \ref{thmdual} implies that, for every $\bm{\mu}\in(\Delta\Omega)^n$, the utility function $U_{\bm{\mu}}(f)=\max_{\Lambda\in\bm{\Lambda}} \min_{\lambda\in\Lambda}\int_{\Omega} u\circ f d p_{\lambda}(\bm{\mu})$ for some weight-set collection $\bm{\Lambda}$. It is established by (\ref{dualutil}) that for every $\bm{\mu}\in (\Delta\Omega)^n$, the utility representation $U_{\bm{\mu}}$ can be rewritten as
$
U_{\bm{\mu}}(f)=\max_{P\in \mathbb{P}}\min_{p\in P} \int_{\Omega} u\circ f dp, 
$ where $\mathbb{P}=\{P_{\Lambda}(\bm{\mu})\mid \Lambda\in\bm{\Lambda}\}$, which is a DSEU representation. It follows from the main representation theorem of \cite{chandrasekher2022dual} that $\succsim_{\bm{\mu}}$ must satisfy C-Independence.

For the ``if'' side, assume that Weak Commutativity is violated, so the aggregation rule is not a dual-self aggregation rule. As a result, the aggregation functional $I:\mathbb{R}^N\rightarrow\mathbb{R}$ is not constant linear. Therefore, there exists $a\in\mathbb{R}^N$, $c\in\mathbb{R}$ and $\beta\in (0,1)$ such that $I(\beta a+(1-\beta) c1_{n})\neq \beta I(a)+(1-\beta)c$. Assume $\ub{a}<\bar{a}$.\footnote{If $\ub{a}=\bar{a}$, \textit{i.e.} $a=\ub{a}1_n$, then by Lemma \ref{a3}, $I(\ub{a} 1_n)=\ub{a}, I(c1_{n})=c$ and $I(\beta \ub{a}1_{n}+(1-\beta) c1_{n})=\beta \ub{a}+(1-\beta)c$, which contradicts to the assumption that  $I(\beta a+(1-\beta) c1_{n})\neq \beta I(a)+(1-\beta)c$.} I construct $\bm{\mu}^*\in (\Delta\Omega)^n$ such that $\mu^*_i=\left(\frac{a_i-\ub{a}}{\bar{a}-\ub{a}}, \frac{\bar{a}-a_i}{\bar{a}-\ub{a}}, 0,0,...\right)$ for all $i\in N$. I construct $f^*\in\mathcal{F}$ be such that $u_{f^*}=(\bar{a},\ub{a},0,0,...)$ and $x^*\in X$ such that $u(x^*)=c$. In this case, by how $\bm{\mu}^*$, $f^*$ and $x^*$ are constructed, $\bm{\mu}^*\cdot u_{f^*}=a$ and $\bm{\mu}^*\cdot u_{x^*}=c1_n$, so that $U_{\bm{\mu}^*}(f^*)=a$ and $U_{\bm{\mu}^*}(x^*)=c$. This leads to $U_{\bm{\mu}}(\beta f^*+(1-\beta)x^*)=I(\beta a+(1-\alpha)c1_{n})\neq \beta I(a)+(1-\beta)c=\beta U_{\bm{\mu}^*}(f^*)+(1-\beta)U_{\bm{\mu}^*}(x^*)$. Finally, let $y^*\in X$ be such that $U_{\bm{\mu}^*}(f^*)=u(y^*)$, then it follows $f^*\sim_{\bm{\mu}^*} y^*$ but $\beta f^*+(1-\beta)x^*\not{\sim}_{\bm{\mu}^*} \beta y^*+(1-\beta)x^*$, which implies $\succsim_{\bm{\mu}^*}$ violates C-Independence.

\subsection{Proof of Proposition \ref{prop41}}

Consider a regular aggregation rule that satisfies Pessimism to Update-then-Aggregate. Regularity implies the existence of an aggregation functional $I:\mathbb{R}^N\rightarrow \mathbb{R}$ such that $U_{\bm{\mu}}(f)=I(\bm{\mu}\cdot u_f)$ for every $\bm{\mu}\in (\Delta\Omega)^n$ and $f\in\mathcal{F}$. Now I want to show $I(\cdot)$ is concave, that is, for every $a,b\in\mathbb{R}^N$ and $\beta\in (0,1)$, $I(\beta a+(1-\beta) b)\geq \beta I(a)+(1-\beta) I(b)$. Recall that Proposition \ref{prop41} requires $|\Omega|\geq 4$. Take $G\subset \Omega$ and $H=G^c$ such that $\{1,2\}\subset G$ and $\{3,4\}\subset H$. 

\textbf{Case 1}: $\ub{a}<\bar{a}$ and $\ub{b}<\bar{b}$. In this case, I construct the acts $f^*$ and suggestion profile $\bm{\mu}^*$ such that $u_{f^*}=(\bar{a},\ub{a},\bar{b},\ub{b},0,...,0)$ and $\mu_i^*=\left(\beta\frac{a_i-\ub{a}}{\bar{a}-\ub{a}}, \beta\frac{\bar{a}-a_i}{\bar{a}-\ub{a}}, (1-\beta)\frac{b_i-\ub{b}}{\bar{b}-\ub{b}}, (1-\beta)\frac{\bar{b}-b_i}{\bar{b}-\ub{b}},0,...,0\right)$ for every $i\in N$. By how $\bm{\mu}^*$ and $f^*$ are constructed, $\bm{\mu}^*\cdot u_{f^*}=\beta a+(1-\beta)b$, ${\bm{\mu}^*}^G\cdot u_{f^*}=a$ and ${\bm{\mu}^*}^H\cdot u_{f^*}=b$. Now, take $x^*,y^*\in X$ such that $f^*\sim_{{\bm{\mu}^*}^G} x^*$ and $f^*\sim_{{\bm{\mu}^*}^H}y^*$, that is, $U_{{\bm{\mu}^*}^G}(f^*)=u(x^*)$ and $U_{{\bm{\mu}^*}^H}(f^*)=u(y^*)$. The existence of such $x^*$ and $y^*$ are guaranteed by the continuity of the preference relations. In this case, Pessimism to Update-then-Aggregate implies $f^*\succsim_{{\bm{\mu}^*}}^G x^*$ and  $f^*\succsim_{{\bm{\mu}^*}}^H y^*$. In this case, the definition of conditional preferences imply $f^*Gy^*\succsim_{\bm{\mu}^*}x^*Gy^*$ and $f^*\succsim_{\bm{\mu}^*}f^*Gy^*$. By transitivity, it holds that $f^*\succsim_{\bm{\mu}^*} x^* G y^*$, i.e. $U_{{\bm{\mu}}^*}(f^*)\geq \beta u(x^*)+(1-\beta) u(y^*)=\beta U_{{\bm{\mu}^*}^G}(f^*)+(1-\beta)U_{{\bm{\mu}^*}^H}(f^*)$. The last equality holds due to the definition of $x^*$ and $y^*$. By how $\bm{\mu}^*$ and $f^*$ are constructed, $U_{{\bm{\mu}}^*}(f^*)=I(\beta a+(1-\beta)b)$, $U_{{\bm{\mu}^*}^G}(f^*)=I(a)$ and $U_{{\bm{\mu}^*}^H}(f^*)=I(b)$. As a result, $I(\beta a+(1-\beta)b)\geq \beta I(a)+(1-\beta)I(b)$, which is the desired result. 

\textbf{Case 2}: $\ub{a}=\bar{a}$ and $\ub{b}<\bar{b}$. In this case, the $\bm{\mu}^*$ constructed above is not well-defined. Alternatively, I construct $\hat{f}\in\mathcal{F}$ and $\hat{\bm{\mu}}\in (\Delta\Omega)^n$ such that $\hat{f}=f^*$ and $$\hat{\mu}_i=\left(\beta, 0, (1-\beta)\frac{b_i-\ub{b}}{\bar{b}-\ub{b}},(1-\beta)\frac{\bar{b}-b_i}{\bar{b}-\ub{b}},0,...,0\right).$$
By how $\hat{\bm{\mu}}$ and $\hat{f}$ are constructed, $\hat{\bm{\mu}}\cdot u_{\hat{f}}=\beta a+(1-\beta)b$, $\hat{\bm{\mu}}^G\cdot u_{\hat{f}}=a$ and ${\hat{\bm{\mu}}}^H\cdot u_{\hat{f}}=b$. Repeat the analysis in Case 1 and we have $I(\beta a+(1-\beta)b)\geq \beta I(a)+(1-\beta)I(b)$. The situations where $\ub{b}=\bar{b}$ can be proved similarly. 

As a result of the analysis above, the aggregation functional $I:\mathbb{R}^N\rightarrow \mathbb{R}$ is concave. For every $\bm{\mu}\in (\Delta\Omega)^N$, every $f,g\in\mathcal{F}$ and $\alpha\in (0,1)$, if $f\sim_{\bm{\mu}}g$, then $I(\bm{\mu}\cdot u_f)=I(\bm{\mu}\cdot u_g)$, and the concavity of $I(\cdot)$ implies $I(\bm{\mu}\cdot u_{\alpha f+(1-\alpha)g})=I(\alpha \bm{\mu}\cdot u_{f}+(1-\alpha)\bm{\mu}\cdot u_g)\geq \alpha I(\bm{\mu}\cdot u_f)+(1-\alpha)I(\bm{\mu}\cdot u_g)=(\bm{\mu}\cdot u_f)$. Therefore, $\alpha f+(1-\alpha)g\succsim_{\bm{\mu}} f$ and the preference relation $\succsim_{\bm{\mu}}$ satisfies ambiguity aversion. 

\subsection{Proof of Corollary \ref{coro42}}

There are four statements in Corollary \ref{coro42}, and I want to show they are equivalent to each other. First, Theorem \ref{ambaverse} establishes the equivalence between statements 2 (\textit{Weak Commutativity} and \textit{Ambiguity Aversion}) and 3 (\textit{Multiple-weight aggregation rule}). Proposition \ref{prop41} above establishes that statement 1 (\textit{Weak Commutativity} and \textit{Pessimism to Update-then-Aggregate}) implies statement 2. Next, we show that statement 3 implies statement 1. That is, if $\sigma$ is a multiple-weight aggregation rule, then Weak Commutativity and Pessimism to Update-then-Aggregate hold. First, since $\sigma$ is a multiple-weight rule, it must be a dual-self aggregation rule. By Theorem \ref{thmdual}, Weak Commutativity must hold. Now it suffices to show that Pessimism to Update-then-Aggregate holds. Consider an arbitrary suggestion profile $\bm{\mu}\in (\Delta\Omega)^N$ and event $E\in\mathcal{E}(\bm{\mu})$ such that $\mu_1(E)=\mu_2(E)=...=\mu_n(E)=\alpha\in (0,1)$. I want to show that $\succsim_{\bm{\mu}^E}$ is more pessimistic than $\succsim_{\bm{\mu}}^E$.

By statement 3, $\sigma$ is a multiple-weight aggregation rule, therefore, there exists a closed and convex set of weights $\Lambda\subset \Delta N$ such that for every $\bm{\mu}\in (\Delta\Omega)^n$, it holds that $U_{\bm{\mu}}(f)=\min_{\lambda\in \Lambda} u\circ f dp_{\lambda}(\bm{\mu})$. As a result,  $U_{\bm{\mu}}(f)=\min_{p\in P_{\Lambda}(\bm{\mu})} u\circ f dp$ (recall from Section \ref{b1} that $P_{\Lambda}(\bm{\mu}) =\{p_{\lambda}(\bm{\mu}) \mid \lambda\in\Lambda\}$). Since $\mu_i(E)=\alpha$ for every $i\in N$, for every $p\in P_{\Lambda}(\bm{\mu})$, it holds that $p=\alpha p^E+(1-\alpha)p^{E^c}$. By the definition of $P_{\Lambda}(\cdot)$, we have $p^E\in P_{\Lambda}(\bm{\mu}^E)$ and $p^{E^c}\in P_{\Lambda}(\bm{\mu}^{E^c})$, as a result, $p\in \alpha P_{\Lambda}(\bm{\mu}^E)+(1-\alpha)P_{\Lambda}(\bm{\mu}^{E^c})$. Therefore, $P_{\Lambda}(\bm{\mu})\subset \alpha P_{\Lambda}(\bm{\mu}^E)+(1-\alpha)P_{\Lambda}(\bm{\mu}^{E^c})$. 

Now, take $f\in\mathcal{F}$, $x\in X$ and suppose $f\succsim_{\bm{\mu}^E}x$, that is, $\min_{p\in P_{\Lambda}(\bm{\mu}^E)}EU_{p}(f)\geq u(x)$. I want to show $f\succsim_{\bm{\mu}}^E x$. By the definition of conditional preferences, it suffices to show that, for every $g\in \mathcal{F}$, $fEg\succsim_{\bm{\mu}} xEg$, i.e. $\min_{p\in P_{\Lambda}(\bm{\mu})}EU_{p}[fEg]\geq \min_{p\in P_{\Lambda}(\bm{\mu})}EU_{p}[xEg]$. Notice that 
\begin{align*}
\min_{p\in P_{\Lambda}(\bm{\mu})}EU_{p}[fEg] &\geq \min_{p\in\alpha P_{\Lambda}(\bm{\mu}^E)+(1-\alpha)P_{\Lambda}(\bm{\mu}^{E^c})}EU_{p}[fEg]\\
&= \alpha \min_{p\in P_{\Lambda}(\bm{\mu}^E)}EU_{p}[fEg]+(1-\alpha)\min_{p\in P_{\Lambda}(\bm{\mu}^{E^c})}EU_{p}[fEg]\\
&=\alpha \min_{p\in P_{\Lambda}(\bm{\mu}^E)}EU_{p}[f]+(1-\alpha)\min_{p\in P_{\Lambda}(\bm{\mu}^{E^c})}EU_{p}[g]\\
&=\alpha u(x)+(1-\alpha)\min_{p\in P_{\Lambda}(\bm{\mu}^{E^c})}EU_{p}[g]
=\min_{p\in P_{\Lambda}(\bm{\mu})}[xEg].
\end{align*}
The first inequality above holds since $P_{\Lambda}(\bm{\mu})\subset \alpha P_{\Lambda}(\bm{\mu}^E)+(1-\alpha)P_{\Lambda}(\bm{\mu}^{E^c})$. As a result of the (in)equalities above, $fEg\succsim_{\bm{\mu}} xEg$. Since $g$ is arbitrarily chosen, it holds that $f\succsim_{\bm{\mu}}^E x$, the desired result. Since $f$ and $x$ are arbitrarily chosen, $\succsim_{\bm{\mu}^E}$ is more pessimistic than $\succsim_{\bm{\mu}}^E$. 

The analysis above establishes the equivalence between statements 1-3 in Corollary \ref{coro42}. Statement 4 states that $\succsim_{\bm{\mu}}$ has an MEU representation for every $\bm{\mu}\in (\Delta\Omega)^n$. Now, I show that statement 2 is equivalent to statement 4. Statement 2 asserts that Weak Commutativity holds and $\succsim_{\bm{\mu}}$ is ambiguity-averse for every $\bm{\mu}$. First, Weak Commutativity is equivalent to $\succsim_{\bm{\mu}}$ satisfies C-Independence by Theorem \ref{thm41}. As a result, Statement 2 is equivalent to that $\succsim_{\bm{\mu}}$ satisfies both C-Independence and ambiguity aversion for every $\bm{\mu}$. Since $\succsim_{\bm{\mu}}$ is a continuous and monotone weak order, C-Independence and ambiguity aversion are jointly equivalent to $\succsim_{\bm{\mu}}$ having an MEU representation, which is statement 4.

\subsection{Proof of Theorem \ref{thm42}} \label{b4}

For the ``only if'' side, consider a regular aggregation rule that satisfies Moderate Commutativity. I want to show that $\succsim_{\bm{\mu}}$ satisfies Independence for every $\bm{\mu}\in (\Delta\Omega)^n$.  First, I show that for every $a,b\in\mathbb{R}^N$ and $\beta\in (0,1)$, $I(\beta a+(1-\beta)b)=\beta I(a)+(1-\beta) I(b)$. Without loss of generality, assume $\ub{a}<\bar{a}$ and $\ub{b}<\bar{b}$.\footnote{Moderate Commutativity implies Weak Commutativity, which implies that $I:\mathbb{R}^N\rightarrow \mathbb{R}$ is constant linear. Since $I$ is constant linear, if $\ub{a}=\bar{a}$, then $a=\ub{a}1_n$, and $I(\beta a+(1-\beta)b)=\beta I(a)+(1-\beta)I(b)$ since $I:\mathbb{R}^N\rightarrow \mathbb{R}$. This is the desired result. The same analysis works for the case where $\ub{b}=\bar{b}$.} Take events $G,H$ such that $H=G^c$, $\{1,2\}\subset G$ and $\{3,4\}\subset H$. Construct $f^*$ such that $u_{f^*}=(\bar{a},\ub{a},\bar{b},\ub{b},0,...,0)$. I also construct $\bm{\mu}^*\in(\Delta\Omega)^n$ such that $\mu_i^*=\left(\beta\frac{a_i-\ub{a}}{\bar{a}-\ub{a}}, \beta\frac{\bar{a}-a_i}{\bar{a}-\ub{a}}, (1-\beta)\frac{b_i-\ub{b}}{\bar{b}-\ub{b}}, (1-\beta)\frac{\bar{b}-b_i}{\bar{b}-\ub{b}},0,...,0\right)$ for every $i\in N$. By how $\bm{\mu}^*$ and $f^*$ are constructed, $\bm{\mu}^*\cdot u_{f^*}=\beta a+(1-\beta)b$, ${\bm{\mu}^*}^G\cdot u_{f^*}=a$ and ${\bm{\mu}^*}^H\cdot u_{f^*}=b$. 

Construct $x^*,y^*\in X$ such that $u(x^*)=I(a)$ and $u(y^*)=I(b)$, then by the definition of the aggregation functional $I$, it holds that $f^*\sim_{{\bm{\mu}^*}^G}x^*$ and $f^*\sim_{{\bm{\mu}^*}^H}y^*$. By how $\bm{\mu}^*$ is constructed, it can be verified that $\mu_i(G)=\beta$ for every $i\in N$. Axiom \ref{moderatecom} (Moderate Commutativity) can therefore be applied for $\bm{\mu}^*$. As a result, $f^*\sim_{{\bm{\mu}^*}}^G x^*$ and $f^*\sim_{{\bm{\mu}^*}}^H y^*$. The definition of conditional preference implies $f^*\sim_{\bm{\mu}^*}x^* Gf \sim_{\bm{\mu}^*}x^* Gy^*$, i.e. $U_{\bm{\mu}^*}(f^*)=U_{\bm{\mu}^*}(x^*G y^*)=\beta u(x^*)+(1-\beta)u(y^*)=\beta I(a)+(1-\beta)I(b)$. Also, notice that the left-hand side of the equation is $U_{\bm{\mu}^*}(f^*)$, which, by the definition of the aggregation functional $I$, is equal to $I(\bm{\mu}^*)(f^*)=I(\beta a+(1-\beta)b)$. Therefore, $I(\beta a+(1-\beta)b)=\beta I(a)+(1-\beta)I(b)$, i.e. $I(\cdot)$ is a mixture linear functional. Since $I$ is also monotonic, there exists a weight $\lambda\in \Delta N$ such that $I(a)=\sum_{i=1}^n \lambda_i a_i$ for every $a\in \mathbb{R}^N$. It follows that $U_{\bm{\mu}}(f)= \int_{\Omega}u\circ f d p_{\lambda}(\bm{\mu})$ for every $\bm{\mu}\in (\Delta\Omega)^n$, which implies $\succsim_{\bm{\mu}}$ satisfying Independence. This completes the proof for the ``only if'' side.

For the ``if'' side, suppose that $\succsim_{\bm{\mu}}$ satisfies Independence for every $\bm{\mu}\in (\Delta\Omega)^n$. For every $a,b\in\mathbb{R}^N$ and $\beta\in (0,1)$, I construct $\bm{\mu}^{**}\in\Delta\Omega$ such that $\mu_i^{**}=\left(\frac{a_i-\ub{a}}{2(\bar{a}-\ub{a})}, \frac{\bar{a}-a_i}{2(\bar{a}-\ub{a})}, \frac{b_i-\ub{b}}{2(\bar{b}-\ub{b})}, \frac{\bar{b}-b_i}{2(\bar{b}-\ub{b})},0,...,0\right)$ for every $i\in N$. Take $f^{**}$ and $g^{**}$ such that $u_{f^{**}}=(2\bar{a},2\ub{a},0,0,0,...)$ and $u_{g^{**}}=(0,0,2\bar{b},2\ub{b},0,...)$. By the way how $\bm{\mu}^{**}$, $f^{**}$ and $g^{**}$ are constructed, it holds that $\bm{\mu}^{**}\cdot u_{f^{**}}=a$, $\bm{\mu}^{**}\cdot u_{g^{**}}=b$ and $\bm{\mu}^{**}\cdot u_{\beta f^{**}+(1-\beta)g^{**}}=\beta a+(1-\beta)b$. Since $\succsim_{\bm{\mu}^{**}}$ satisfies Independence (and all the other Anscombe-Aumann axioms), $U_{\bm{\mu}^{**}}(\cdot)$ is mixture linear, i.e. $\beta U_{\bm{\mu}^{**}}(f^{**})+(1-\beta) U_{\bm{\mu}^{**}}(g^{**})=U_{\bm{\mu}^{**}}(\beta f^{**}+(1-\beta) g^{**})$. As a result, $\beta I(a)+(1-\beta)I(b)=I(\beta a+(1-\beta)b)$. Since $a,b\in \mathbb{R}^N$ and $\beta\in (0,1)$ are arbitrarily chosen, $I:\mathbb{R}^N\rightarrow \mathbb{R}$ is mixture linear, i.e. there exists a weight $\lambda\in \Delta N$ such that $I(a)=\sum_{i=1}^n \lambda_i a_i$ for every $a\in\mathbb{R}^N$. For every $\bm{\mu}\in \Delta\Omega$ and $E\in\mathcal{E}(\bm{\mu})$ such that $\mu_1(E)=\mu_2(E)=...=\mu_n(E)=\alpha\in (0,1)$, it holds that $U_{\bm{\mu}}(f)=EU_{p_{\lambda}(\bm{\mu})}(f)$, where $p_{\lambda}(\bm{\mu})=\sum_{i=1}^n \lambda_i\mu_i$. This implies $\succsim_{\bm{\mu}}^E$ is represented by $EU_{[p_{\lambda}(\bm{\mu})]^E}(f)$. Since $\mu_i(E)=\alpha$ for every $i\in N$, for every $\omega\in E$, $[p_{\lambda}(\bm{\mu})]^E(\omega)=p_{\lambda}(\bm{\mu})(\omega)/\alpha=\sum_{i=1}^n\lambda_i \mu_i(\omega)/\alpha=\sum_{i=1}^n\lambda_i \mu_i^E(\omega)=p_{\lambda}(\bm{\mu}^E)(\omega)$. In conclusion, $[p_{\lambda}(\bm{\mu})]^E=p_{\lambda}(\bm{\mu}^E)$. herefore, $\succsim_{\bm{\mu}}^E$ and $\succsim_{\bm{\mu}^E}$ are represented by SEU representations with identical beliefs. That is, the two preferences are identical, and Moderate Commutativity holds.

\subsection{Proof of Corollary \ref{coro43}}

 First, the equivalence between statements 1 (\textit{Moderate Commutativity}) and 3 (\textit{Independence}) is established by Theorem \ref{thm42}. Second, the equivalence between statements 2 (\textit{Weak Commutativity} and \textit{STP}) and 4 (\textit{Linear Aggregation}) is established by Theorem \ref{ambfree}. Third, that statement 2 (\textit{Independence}) implies the mixture linearity of the aggregation functional $I:\mathbb{R}^N\rightarrow\mathbb{R}$ is established by the last paragraph in Section \ref{b4}, while the mixture linearity of $I$ implies statement 4 (\textit{Linear Aggregation}). Finally, statement 4 (\textit{Linear Aggregation}) implies the DM's preference $\succsim_{\bm{\mu}}$ always has a SEU representation, which implies statement 2 (\textit{Independence}). As a result, statements 1-4 are equivalent to each other.

\subsection{Proof of Theorem \ref{thmimposs}}

The ``if'' side of Theorem \ref{thmimposs} is trivial. For the ``only if'' side, consider a regular aggregation rule that satisfies Commutativity, which implies Moderate Commutativity. Therefore, by Theorem 4.2, there exists a weight $\lambda\in \Delta N$ such that for every $\bm{\mu}\in (\Delta\Omega)^n$ and $f\in \mathcal{F}$, 
$
U_{\bm{\mu}}(f)=EU_{p_{\lambda}(\bm{\mu})} (f),
$ 
where $p_{\lambda}(\bm{\mu})=\sum_{i=1}^n \lambda_i \mu_i$. Now, suppose toward a contradiction that the aggregation rule satisfies non-dictatorship, then without loss of generality, assume $\lambda_1>0$ and $\lambda_2>0$. Take $G\subset \Omega$ such that $\{1,2\}\subset G$. Construct $\bm{\mu}^*$ such that 
$$
\mu^*_i=\begin{cases}
(.1,0,.9,0,0,...)& \text{ if } i=1\\
 (0,1,0,0,...)& \text{ if } i\neq 1,
\end{cases}
$$
then 
$$
{\mu^*_i}^G=\begin{cases}
(1,0,0,0,0,...) &\text{ if } i=1\\
 (0,1,0,0,...)& \text{ if } i\neq 1,
\end{cases}
$$
By how ${\bm{\mu}}^*$ is constructed, according to the distribution $p_{\lambda}({\bm{\mu}^*}^G)$, the probability of $\omega=1$ is equal to $\lambda_1\in (0,1)$. On the other hand, according to the distribution $[p_{\lambda}(\bm{\mu}^*)]^G$, the probability of $\omega=1$ is equal to $\frac{.1\lambda_1}{1-.9\lambda_1}<\lambda_1$. As a result,  $p_{\lambda}({\bm{\mu}^*}^G)\neq [p_{\lambda}({\bm{\mu}^*})]^G$, which implies that $\succsim^G_{\bm{\mu}^*}$ is not equal to $\succsim_{{\bm{\mu}^*}^G}$, a violation of Commutativity. As a result, the aggregation rule must be a dictatorship, which is the desired result.

%% file: Appendix/c.tex
\section{Proof of Lemma \ref{A1}}
\label{appendixc}

Lemma \ref{A1} is closely related to the main result in \cite{hill2011unanimity}. Both results state that two unanimity conditions in the form of (\ref{condition1}) and (\ref{condition2}) imply the stronger unanimity condition in the form of (\ref{main}). The key difference is this paper does not put restrictions to the function $\varsigma(T,\cdot):\mathbb{R}^{l}\rightarrow \mathbb{R}$, while \cite{hill2011unanimity} consider the class of constant linear and concave $\varsigma(T,\cdot)$. To start the proof, let $\varsigma:\mathcal{T}\times\mathbb{R}^l\rightarrow \mathbb{R}$ be a function such that (\ref{condition1}) and (\ref{condition2}) are satisfied. That is, $T  v=T  w$ implies $\varsigma(T,v)=\varsigma(T,w)$ and $T  v=T'  v$ implies $\varsigma(T,v)=\varsigma(T',v)$. In order to show (\ref{main}), I want to show that for $T,T'\in\mathcal{T}$ and $v,w\in\mathbb{R}^l$ such that $T  v=T'  w$, it holds that $\varsigma(T,v)=\varsigma(T',w)$.

First, suppose $v= b1_{l}$ where $b$ is a real number. In that case, $Tv=T'v=b 1_n$, which, together with $Tv=T'w$, implies $Tv=T'v=T'w$. Applying (\ref{condition1}) and (\ref{condition2}) implies  $\varsigma(T,v)=\varsigma(T',v)=\varsigma(T',w)$, which is the desired result. In the rest of the proof, I assume $v$ and $w$ are not constant vectors, that is, $v,w\neq b1_l$ for every $b\in\mathbb{R}$. Recall that $\bar{v}\equiv \max_{1\leq j\leq l}v_j$ and $\ub{v}\equiv\min_{1\leq j\leq l}v_j$; since $v$ is non-constant, it must holds that $\bar{v}>\ub{v}$. The next lemma shows (\ref{main}) is satisfied if the vectors $v,w$ satisfy a specific restriction. For any $j=1,...,l$, denote $e_j$ as the $l$-dimensional vector with the $j$-th entry equal to 1 and all other entries equal to 0. 

\begin{lemma}
\label{lemmac2}
If for some $j, j'\in \{1,...,l\}, k,k'>0$ and $b,b'\in\mathbb{R}$, it holds that $v=k e_{j}+b 1_{l}$, $w=k' e_{j'}+b' 1_{l}$ , then 
\begin{align*}
T  v=T'  w\implies \varsigma(T,v)=\varsigma(T',w).
\end{align*}
\end{lemma}

\begin{proof}

Before starting the proof, I first introduce some notations. Recall that I use $T_{n\times l}$ to define a (real-valued) matrix with $n$ rows and $l$ columns. For $j=1,2,...,l$, I use $T_j$ to denote the $j$-th column of $T_{n\times l}$, and $T$ can be written as 
$$
T=\begin{pmatrix}
    T_1& T_2&...&T_l
\end{pmatrix}.
$$
I use $\bm{0}_{n'\times l'}$ to denote the $n'\times l'$ matrix with all the entries being zero. Next, I prove Lemma \ref{lemmac2} in the following two cases. 

\textbf{Case 1}: $j\neq j'$: Without loss of generality, assume $j=1$ and $j'=2$. That is $v=ke_1+b1_l$ and $w=k'e_2+b'1_l$. In this case,
\begin{align}
 T  v=   T  (ke_1+b1_l)&= \begin{pmatrix} T_1 &
    0& 1_{n}-T_1 & \bm{0}_{n\times (l-3)}
    \end{pmatrix}   (ke_1+b1_l)\label{aa}\\
    &= \begin{pmatrix} T_1 &
    0& 1_{n}-T_1 & \bm{0}_{n\times (l-3)}
    \end{pmatrix}   (ke_1+ke_2+b1_l)\\
    &= \begin{pmatrix} \frac{T_1}{2} & \frac{T_1}{2} & 1_{n}-T_1 & \bm{0}_{n\times (l-3)}
    \end{pmatrix}   (ke_1+ke_2+b1_l)\\
    &=\begin{pmatrix} \frac{T_1}{2} & \frac{T_1}{2} & 1_{n}-T_1 & \bm{0}_{n\times (l-3)}
    \end{pmatrix}   (2 ke_1+b1_l) \label{aaa}\\
    & \label{aaaa}=\begin{pmatrix} \frac{T_1}{2} & \frac{T'_2}{2} & 1_{n}-\frac{T_1}{2}-\frac{T'_2}{2} & \bm{0}_{n\times (l-3)}
    \end{pmatrix}   (2 ke_1+b1_l). 
\end{align}

Substitute $T, k, b$ by $T',k',b'$, and the above chain of equations can be replicated, which generates equation $(\ref{ab})$ below. The assumption that $Tv=T'w$, (\ref{aaaa}), and (\ref{ab}) jointly imply equation $(\ref{ac})$.
\begin{align}
T'  w=T'  (k'e_2+b'1_l)=
&\begin{pmatrix} \frac{T_1}{2} & \frac{T'_2}{2} & 1_{n}-\frac{T_1}{2}-\frac{T'_2}{2} & \bm{0}_{n\times (l-3)}
    \end{pmatrix}   (2 k'e_2+b'1_l)\label{ab}\\
    =
&\begin{pmatrix} \frac{T_1}{2} & \frac{T'_2}{2} & 1_{n}-\frac{T_1}{2}-\frac{T'_2}{2} & \bm{0}_{n\times (l-3)}
    \end{pmatrix}   (2 ke_1+b 1_l). \label{ac}
\end{align}
Apply conditions (\ref{condition1}), (\ref{condition2}) alternatively to (\ref{aa}) - (\ref{ac}), and it follows that 
\begin{align}
\varsigma(T,v)=&\varsigma
\left(\begin{pmatrix} \frac{T_1}{2} & \frac{T'_2}{2} & 1_{n}-\frac{T_1}{2}-\frac{T'_2}{2} & \bm{0}_{n\times (l-3)}
    \end{pmatrix}, 2ke_1+b1_l\right)=\varsigma(T',w),
\end{align}
which completes the proof for Case 1.

\textbf{Case 2}: $j= j'$. Without loss of generality, assume $j=j'=1$. Let $z=ke_2+b1_l$ and $T''=\begin{pmatrix} 
T_2 & T_1 & T_3 & ... & T_l
\end{pmatrix}$. It follows that $Tv=T'w=T''z$. By \textbf{Case 1}, it holds that
$
\varsigma(T,v)=\varsigma(T'',z) \text{ and } \varsigma(T',w)=\varsigma(T'',z)
$. As a result, $\varsigma(T,v)=\varsigma(T',w)$.
\end{proof}
With the preparation of Lemma \ref{lemmac2}, we are ready to prove Lemma \ref{A1}. Fix $v,w$ that are not constant vectors. Suppose $T  v=T'  w$ and I want to show $\varsigma(T,v)=\varsigma(T',w)$. Assume $v_{j_1}=\bar{v}$ and $v_{j_2}=\ub{v}$. Define $T^a\in\mathcal{T}$ such that the $j_1$-th column is $T_{j_1}^a=\frac{T  v-\ub{v} 1_n}{\bar{v}-\ub{v}}$, the $j_2$-th column is $T_{j_2}^a=1_n-T_1^a$, and the remaining columns are zeros.\footnote{Notice that $T v-\ub{v}1_n=T v-\ub{v} T 1_l=T  (v-\ub{v}1_l)\geq 0$. Also, $1_n-[(Tv-\ub{v}1_n)/(\bar{v}-\ub{v})]\geq 1_n-[(T\bar{v}1_l-\ub{v}1_n)/(\bar{v}-\ub{v})]=0$. As a result, all the entries of $T^a$ are indeed (weakly) positive.} By the definition of $T^a$, both $T^a  v$ and $T^{a}  ((\bar{v}-\ub{v})e_{j_1}+\ub{v}1_l)$ are equal to
\begin{align*}
T^a_{j_1}\cdot \bar{v}+T^a_{j_2}\cdot \ub{v}=\frac{T  v-\ub{v} 1_n}{\bar{v}-\ub{v}}\cdot \bar{v}+\left(1_n-\frac{T  v-\ub{v} 1_n}{\bar{v}-\ub{v}}\right)\cdot \ub{v}=Tv.
\end{align*}
That is, $
T  v=T^a  v=T^{a}  ((\bar{v}-\ub{v})e_{j_1}+\ub{v}1_l)
$. As a result, 
\begin{align}
\label{ba}
\varsigma(T,v)=\varsigma(T^a,v)=\varsigma(T^a,(\bar{v}-\ub{v})e_{j_1}+\ub{v}1_l).
\end{align}

Similarly, assume $w_{j_3}=\bar{w}$ and $w_{j_4}=\ub{w}$. Define $T^b\in\mathcal{T}$ such that the $j_3$-th column is $T_{j_3}^b=\frac{T'w-\ub{w}1_n}{\bar{w}-\ub{w}}$, the $j_4$-th column is $T^b_{j_4}=1_n-T_1^a$, and the remaining columns are all zeros. It holds that 
$
T'w=T^bw=T^b((\bar{w}-\ub{w})e_{j_3}+\ub{w}1_l)
$. As a result, 
\begin{align}
\label{bb}
\varsigma(T',w)=\varsigma(T^b,w)=\varsigma(T^b,(\bar{w}-\ub{w})e_{j_3}+\ub{w}1_l).
\end{align}

Notice that $Tv=T'w$, which implies $T^a\cdot [(\bar{v}-\ub{v})e_{j_1}+\ub{v}1_l]=T^b \cdot [(\bar{w}-\ub{w})e_{j_3}+\ub{w}1_l]$. By Lemma \ref{lemmac2}, it holds that $\varsigma(T^a,(\bar{v}-\ub{v})e_{j_1}+\ub{v}1_l)=\varsigma(T^b,(\bar{w}-\ub{w})e_{j_3}+\ub{w}1_l)$, which, by  (\ref{ba}) and (\ref{bb}), implies $\varsigma(T,v)=\varsigma(T',w)$. This finishes the proof for Lemma \ref{A1}. 